\documentclass[11pt,letterpaper]{article}
\usepackage[margin=1in]{geometry}

\usepackage{hyperref}
\usepackage{amsmath,amsthm,amssymb,amsfonts}
\usepackage{mathtools}
\usepackage{cleveref}
\usepackage{xspace}
\usepackage{url}
\usepackage[dvipsnames]{xcolor}
\usepackage[ruled]{algorithm}
\usepackage{algpseudocode}
\usepackage{appendix}
\usepackage{paralist}
\usepackage{enumitem}

\usepackage{booktabs}
\usepackage{thm-restate}
\usepackage{graphicx}
\usepackage[obeyFinal,textsize=footnotesize]{todonotes}

\hypersetup{colorlinks=true,allcolors=blue}

\theoremstyle{plain}
\newtheorem{theorem}{Theorem}[section]
\newtheorem{lemma}[theorem]{Lemma}
\newtheorem{claim}[theorem]{Claim}
\newtheorem*{claim*}{Claim}

\newtheorem{fact}[theorem]{Fact}

\theoremstyle{definition}
\newtheorem{definition}[theorem]{Definition}

\theoremstyle{remark}
\newtheorem{remark}[theorem]{Remark}

\newcommand{\ProblemName}[1]{\textsc{#1}}
\newcommand{\kMedian}{\ProblemName{$k$-Median}\xspace}
\newcommand{\kMeans}{\ProblemName{$k$-Means}\xspace}
\newcommand{\kzC}{\ProblemName{$(k,z)$-Clustering}\xspace}
\newcommand{\pzFL}{\ProblemName{Power-$z$ Facility Location}\xspace}
\newcommand{\FL}{\ProblemName{Facility Location}\xspace}

\newcommand{\Ot}{\ensuremath{\tilde{O}}\xspace}
\newcommand{\eqdef}{\coloneqq}

\DeclareMathOperator{\E}{\mathbb{E}}

\DeclareMathOperator{\wopen}{\mathfrak{f}}
\DeclareMathOperator{\cost}{cost}

\def\RR{{\mathbb{R}}}

\DeclareMathOperator{\poly}{poly}
\DeclareMathOperator{\polylog}{polylog}

\DeclareMathOperator{\fl}{fl}
\DeclareMathOperator{\cl}{cl}
\DeclareMathOperator{\OPT}{OPT}
\DeclareMathOperator{\ALG}{ALG}
\DeclareMathOperator{\dist}{dist}
\DeclareMathOperator{\diam}{diam}

\DeclareMathOperator{\crit}{crit}

\newcommand{\EPone}{\ensuremath{\mathcal{E}_{\mathrm{P1}}}\xspace}
\newcommand{\EPtwo}{\ensuremath{\mathcal{E}_{\mathrm{P2}}}\xspace}
\newcommand{\Econd}{\ensuremath{\mathcal{E}_{\mathrm{cond}}}\xspace}
\newcommand{\Eextend}{\ensuremath{\mathcal{E}_{\mathrm{extend}}}\xspace}

\newcommand{\Pstay}{\ensuremath{P_{\mathrm{stay}}}\xspace}
\newcommand{\Pextend}{\ensuremath{P_{\mathrm{extend}}}\xspace}

\newcommand{\ellstay}{\ensuremath{\ell_{\mathrm{stay}}}\xspace}
\newcommand{\ellextend}{\ensuremath{\ell_{\mathrm{extend}}}\xspace}
\newcommand{\Acond}{\ensuremath{A_{\mathrm{cond}}}\xspace}
\newcommand{\Anew}{\ensuremath{A_{\mathrm{new}}}\xspace}
\newcommand{\Anewgeq}{\ensuremath{A_{\mathrm{new}}^{\geq \hat{r}_{x_t}}}\xspace}
\newcommand{\OPTfl}{\ensuremath{\OPT^{\mathrm{fl}}}\xspace}
\newcommand{\OPTkz}{\ensuremath{\OPT^{\mathrm{cl}}_z}\xspace} \newcommand{\OPTkzGuess}{\ensuremath{\widehat{\OPT}^{\mathrm{cl}}_z}\xspace} 

\newcommand{\algoFLz}{\ensuremath{\ALG^{\mathrm{fl}}_z}\xspace} 

\newcommand{\POconsistent}{\propto}

 \providecommand{\set}[1]{{\{#1\}}}

\let\eps\varepsilon
\let\epsilon\varepsilon

\title{Fully-Scalable MPC Algorithms for Clustering in High Dimension}

\author{
  Artur Czumaj\thanks{
    Partially supported by the Centre for Discrete Mathematics and its Applications (DIMAP), by EPSRC award EP/V01305X/1, by a Weizmann-UK Making Connections Grant, and by an IBM Award.
    Email: \texttt{A.Czumaj@warwick.ac.uk}
  } \\
  University of Warwick
  \and Guichen Gao\thanks{Email: \texttt{gc.gao@stu.pku.edu.cn}} \\
  Peking University
  \and Shaofeng H.-C. Jiang\thanks{
    Research partially supported by a national key R\&D program of China No. 2021YFA1000900,
    a startup fund from Peking University, and the Advanced Institute of Information Technology, Peking University.
    Email: \texttt{shaofeng.jiang@pku.edu.cn}
  } \\
  Peking University
  \and Robert Krauthgamer\thanks{
    Partially supported
      by the Israel Science Foundation grant \#1336/23,
      by a Weizmann-UK Making Connections Grant,
      by the Israeli Council for Higher Education (CHE) via the Weizmann Data Science Research Center,
      and by a research grant from the Estate of Harry Schutzman.
      Email: \texttt{robert.krauthgamer@weizmann.ac.il}
  } \\
  Weizmann Institute of Science
  \and Pavel Vesel\'y\thanks{
    Partially supported by GA ČR project 22-22997S and by Center for Foundations of Modern Computer Science (Charles Univ. project UNCE 24/SCI/008).
    Email: \texttt{vesely@iuuk.mff.cuni.cz}
  } \\
  Charles University
}

\begin{document}

\maketitle

\setcounter{page}{0}
\thispagestyle{empty}

\begin{abstract}
We design new parallel algorithms for clustering in high-dimensional Euclidean spaces.
These algorithms run in the Massively Parallel Computation (MPC) model,
and are \emph{fully scalable},
meaning that the local memory in each machine may be $n^{\sigma}$
for arbitrarily small fixed $\sigma>0$.
Importantly, the local memory may be substantially smaller than
the number of clusters $k$,
yet all our algorithms are fast, i.e., run in $O(1)$ rounds.

We first devise a fast MPC algorithm for $O(1)$-approximation of uniform \FL.
This is the first fully-scalable MPC algorithm
that achieves $O(1)$-approximation for any clustering problem in general geometric setting;
previous algorithms only provide $\poly(\log n)$-approximation
or apply to restricted inputs, like low dimension or small number of clusters $k$;
e.g.\ [Bhaskara and Wijewardena, ICML'18; Cohen-Addad et al., NeurIPS'21; Cohen-Addad et al., ICML'22].
We then build on this \FL result and devise a fast MPC algorithm
that achieves $O(1)$-bicriteria approximation for \kMedian and for \kMeans,
namely, it computes $(1+\varepsilon)k$ clusters of cost within $O(1/\varepsilon^2)$-factor of the optimum for $k$ clusters.

A primary technical tool that we introduce, and may be of independent interest,
is a new MPC primitive for geometric aggregation,
namely, computing for every data point a statistic of its approximate neighborhood,
for statistics like range counting and nearest-neighbor search.
Our implementation of this primitive works in high dimension,
and is based on consistent hashing (aka sparse partition),
a technique that was recently used for streaming algorithms [Czumaj et al., FOCS'22].
\end{abstract}

\newpage

\section{Introduction}
\label{sec:intro}

Clustering large data sets is a fundamental computational task that
has been studied extensively due to its wide applicability in data science,
including, e.g., unsupervised learning, classification, data mining.
Two highly popular and extremely basic problems are \kMedian and \kMeans.
In the geometric (Euclidean) setting,
the \kMedian problem asks,
given as input an integer $k$ and a set $P \subseteq \RR^d$ of $n$ data points,
to compute a set $C\subset \RR^d$ of $k$ center points,
so as to minimize the sum of distances from each point in $P$ to its nearest center.
The \kMeans problem is similar except that
instead of the sum of distances, one minimizes the sum of squares of distances.
A closely related problem is (uniform) \FL,
which can be viewed as a Lagrangian relaxation of \kMedian,
where the number of centers can vary but it adds a penalty to the objective,
namely, each center (called open facility) incurs a given cost $\wopen>0$.
Other variants include a similar relaxation of \kMeans,
or generalizing the squaring each distance to raising it to power $z\ge1$
(see \Cref{sec:prelim} for formal definitions),
and there is a vast literature on these computational problems.

Clustering is often performed on massive datasets,
and it is therefore important
to study clustering
in the framework of distributed and parallel computations.
We consider
fundamental clustering problems
in the theoretical model of \emph{Massively Parallel Computation (MPC)},
which captures key aspects of modern large-scale computation systems,
such as MapReduce, Hadoop, Dryad, or Spark.
The MPC model was introduced over a decade ago by Karloff, Suri, and Vassilvitskii~\cite{KSV10},
and over time has become
the standard theoretical model to study data-intensive parallel algorithms (see, e.g., \cite{BKS17,DBLP:conf/isaac/GoodrichSZ11,IKLMV23}).
At a high level, the MPC model consists of many machines
that communicate with each other synchronously in a constrained manner,
in order to solve a desired task in a few rounds.
In more detail, an MPC system has $m$ \emph{machines},
each with a \emph{local memory} of $s$ words,
hence the system's \emph{total memory} is the product $m \cdot s$.
Computation takes place in synchronous rounds,
where machines perform arbitrary computations on their local memory
and then exchange messages with other machines.
Every machine is constrained to send and receive at most $s$ words in every round,
and every message must be destined to a single machine (not a broadcast).
At each round, every machine
processes its incoming messages to generate its outgoing messages,
usually without any computational restrictions on this processing (e.g., running time).
At the end of the computation, the machines collectively output the solution.
The efficiency of an algorithm is measured by the \emph{number of rounds} and by the local-memory size $s$.
The total space should be low as well, but this is typically of secondary importance.

The local-memory size $s$ is a key parameter and should be sublinear in the input size $N$
(for if $s \ge N$ then any sequential algorithm can be executed locally by a single machine in one round).
We focus on the more challenging regime, called \emph{fully-scalable MPC},
where the local-memory size is an arbitrarily small polynomial,
i.e., $s = N^{\sigma}$ for any fixed $\sigma\in(0,1)$.
This regime is highly desirable because in practice,
the local memory of the machines is limited by the hardware available at hand,
and the parameter $\sigma$ can be used to model this limitation.
The total memory should obviously be large enough to fit the input,
i.e., $m \cdot s \geq N$, and ideally not much larger.

Modern research on MPC algorithms
aims to solve fundamental problems in as few rounds as possible,
ideally in $O(1)$ rounds, even in the fully-scalable regime,
and so far there have been some successes in designing fast (i.e., $O(1)$-rounds) algorithms for fixed $\sigma < 1$;
see e.g., \cite{DBLP:conf/stoc/AndoniNOY14,BCMT22,DBLP:conf/isaac/GoodrichSZ11}.
In contrast, many graph problems, including fundamental ones
such as connectivity or even distinguishing between $1$ and $2$ cycles,
seem to require super-constant number of rounds
(at the same time, proving a super-constant lower bound
would yield a breakthrough in circuit complexity~\cite{RVW18}).
One can also consider smaller local space $s$, i.e., $\sigma = o(1)$,
and in fact many algorithms work as long as $s \ge \polylog(N)$.
In this regime, the best one can hope for is usually not $O(1)$ but rather $O(\log_s N)$ rounds,
because even broadcasting a single number to all machines requires
$\Omega(\log_s N)$ rounds; see~\cite{RVW18} for further discussion.

Clustering of massive datasets has received significant attention recently,
with numerous algorithms that span multiple computational models,
such as streaming~\cite{DBLP:conf/icml/BravermanFLSY17,DBLP:journals/corr/abs-1802-00459} or
distributed and parallel models~\cite{BT10,EIM11,DBLP:conf/spaa/BlellochGT12,DBLP:journals/pvldb/BahmaniMVKV12,DBLP:conf/nips/BalcanEL13,DBLP:conf/nips/MalkomesKCWM15,DBLP:conf/kdd/BachemL018,DBLP:journals/pvldb/CeccarelloPP19,BCMT22,dBBM23}.
However, this advance is facing a barrier,
first identified by Bhaskara and Wijewardena~\cite{DBLP:conf/icml/BhaskaraW18}:
State-of-the-art methods for $k$-clustering
(i.e., when the number $k$ of clusters is specified in the input)
typically fail to yield fully-scalable MPC algorithms,
because when implemented in the MPC model,
these methods usually require local-memory size $s\geq \Omega(k)$.
For example, many streaming algorithms for $k$-clustering problems
are based on a linear sketch and thus readily applicable to the MPC model;
however, all known streaming algorithms require $\Omega(k)$ space,
and in fact there are lower bounds to this effect \cite{CSWZ16} even in low dimension.
Another example is coresets~\cite{DBLP:conf/stoc/Har-PeledM04},
an extremely effective method to decrease the size of the dataset, even in high dimension;
see~\cite{DBLP:conf/stoc/Cohen-AddadSS21,DBLP:conf/focs/BravermanCJKST022,DBLP:conf/stoc/Cohen-AddadLSS22} for the latest in this long line of research.
Coresets can often be merged and/or applied repeatedly via the so-called ``merge-and-reduce'' framework~\cite{DBLP:conf/stoc/Har-PeledM04},
and can thus be applied successfully in many different settings,
including the parallel setting.
However, this method suffers from the drawback (which so far seems inherent)
that each coreset must be stored in its entirety in a single machine,
and it is easy to see that a coreset for $k$-clustering must be of size at least $k$;
see also~\cite{DBLP:conf/stoc/Cohen-AddadLSS22}.

These shortcomings have led to a surge of interest in MPC algorithms for clustering,
with emphasis on fully-scalable and fast algorithms
(i.e., taking $O(1)$ rounds even for large $k$)
that work in high dimension $d$.
The first such result for $k$-clustering was by Bhaskara and Wijewardena \cite{DBLP:conf/icml/BhaskaraW18};
their $O(1)$-round fully-scalable MPC algorithm
achieves polylogarithmic bicriteria approximation for \kMeans,
i.e., it outputs $k\cdot\polylog(n)$ centers (or clusters)
whose cost is within $\polylog(n)$-factor of the optimum for $k$ centers.
Recent work by Cohen-Addad et al.~\cite{DBLP:conf/nips/Cohen-AddadLNSS21}
achieves for \kMedian a true $\polylog(n)$-approximation
(i.e., without violating the bound $k$ on number of centers);
their algorithm actually computes a hierarchical clustering,
that is, a single hierarchical structure of centers
that induces an approximate solution for \kMedian for every $k$.
These were the best fully-scalable $O(1)$-round MPC algorithms
known for \kMedian and \kMeans prior to our work,
and it remained open to achieve
better than $O(\log n)$-approximation,
even as a bicriteria approximation.

Not surprisingly, previous work on MPC algorithms
has focused also on achieving improved approximation for restricted inputs.
In particular, $(1+\eps)$-approximation is achieved in~\cite{CMZ22}
for \kMeans and \kMedian on inputs that are \emph{perturbation-resilient} \cite{BL12},
meaning that perturbing pairwise distances by (bounded) multiplicative factors
does not change the optimal clusters;
unfortunately, results for such special cases rarely generalize to all inputs.
There are also known fully-scalable MPC algorithms for $k$-center
\cite{DBLP:conf/aaai/BateniEFM21, DBLP:conf/spaa/CoyCM23},
but they are applicable only in low dimension $d$ and thus less relevant to our focus.

Despite these many advances on fully-scalable MPC algorithms,
we are not aware of $O(1)$-approximation
\emph{for any clustering problem} in high dimension.\footnote{Having said that, $(3+\eps)$-approximation is known for (non-geometric) correlation clustering \cite{BCMT22}.}
From a technical perspective, the primary difficulty
is to distribute the data points across machines,
so that points arriving to the same machine are from the same cluster.
This task seems to require knowing the clustering,
which is actually the desired output, so we run into a chicken-and-egg problem!
Fortunately, powerful algorithmic tools can spot points that are close together,
even in high-dimensional geometry.
Namely, locality sensitive hashing (LSH) is employed in~\cite{DBLP:conf/icml/BhaskaraW18}
and tree embedding is used in~\cite{DBLP:conf/nips/Cohen-AddadLNSS21},
but as mentioned above, these previous results do not achieve $O(1)$-approximation.
We rely instead on the framework of \emph{consistent hashing},
which was first employed for streaming algorithms in~\cite{arxiv.2204.02095},
although it was originally proposed in~\cite{DBLP:conf/stoc/JiaLNRS05};
see \Cref{sec:techniques} for details.

\subsection{Our Results}
\label{sec:results}

We devise fully-scalable $O(1)$-round MPC algorithms for a range of clustering problems in $\RR^d$, most notably \FL, \kMedian, and \kMeans.
We first devise an $O(1)$-approximation algorithm for \FL and then exploit the connection between the problems to solve $k$-clustering.

By convention, we express memory bounds in machine words, where each word
can store a counter in the range $[\poly(n)]$, i.e., $O(\log n)$ bits,
and/or a coordinate of a point from $\RR^d$ with restricted precision
(comparable to that of a point from $P$).
For example, the input $P\subset \RR^d, |P|=n$ fits in $nd$ words.
Throughout, the notation $O_\eps(\cdot)$ hides factors that depend on $\eps$.
We stress that in the next theorem,
the memory bounds are polynomial in the dimension $d$ (and not exponential);
this is crucial because the most important case is $d=O(\log n)$,
as explained in \Cref{remark:details}\ref{it:dim_reduction} below.

\begin{theorem}[Simplified version; see \Cref{thm:ufl}]
\label{thm:ufl_intro}
Let $\eps,\sigma \in (0,1)$ be fixed.
There is a randomized fully-scalable MPC algorithm that,
given a multiset $P\subset\RR^d$ of $n$ points
distributed across machines with local memory of size $s \ge n^{\sigma}\cdot\poly(d)$,
computes in $O_\sigma(1)$ rounds an $O_\eps(1)$-approximation  for uniform \FL.
The algorithm uses $O(n^{1+\eps})\cdot \poly(d)$ total space.
\end{theorem}

\begin{remark} \label{remark:details}
This simplified statement omits standard technical details:
\begin{enumerate}[label=(\alph*),nosep]
\item \label{it:fixed_sigma}
    We consider here fixed $\sigma$,
    but the algorithm works for any local-memory size $s \ge \poly(d \log n)$,
    and in that case the number of rounds becomes $O(\log_sn)$.
  \item
    Similarly, $\eps$ can be part of the input,
    and the approximation factor is in fact $(1/\eps)^{O(1)}$.
  \item
    The algorithm succeeds with high probability $1-1/\poly(n)$.
\item
    The input $P$ can be a multiset,
    in contrast to streaming algorithms for such problems~\cite{Indyk04, arxiv.2204.02095},
    where data points must be distinct.
  \item
    The algorithm outputs a feasible solution,
    i.e., a set of facilities $F\subset\RR^d$
    and an assignment of the input points to facilities.
    In fact, in all our algorithms $F\subseteq P$.
    We focus throughout on computing $F$,
    since our methods can easily compute a near-optimal assignment given $F$.
  \item \label{it:dim_reduction}
    We state here the dependence on $d$ explicitly for sake of completeness,
    but our result works whenever $s \geq \polylog (dn)$,
    which is preferable when $d \gg \log n$.
    This is achieved by reducing to the case $d=O(\log n)$,
    using a standard dimension reduction, as discussed in \Cref{remark:jl}.
\end{enumerate}
\end{remark}

Our full result (in \Cref{thm:ufl}) is more general in two respects.
First, it addresses a generalization of \FL where distances are raised to power $z\ge 1$
(see \Cref{sec:prelim} for definitions).
Second, it assumes access to consistent hashing with parameters $\Gamma$ and $\Lambda$,
and here we plugged in a known construction \cite{arxiv.2204.02095}
that achieves a tradeoff $\Gamma = O(1/\eps)$ and $\Lambda = n^\eps$
for any desired $\eps\in(0,1)$;
see \Cref{lemma:hash_bound} and \Cref{remark:jl}.
The description of the hash function should be included in our MPC algorithm,
but it takes only $\poly(d)$ bits, which is easily absorbed in our bounds.

Previously, no fully-scalable MPC algorithm was known for \FL,
although one can apply the algorithm for \kMedian from~\cite{DBLP:conf/nips/Cohen-AddadLNSS21}
to obtain an $O(\polylog(n))$-approximation in $O(\log_s n)$ rounds.
Another previous result that one could use is a streaming algorithm
that achieves $O(1)$-approximation~\cite{arxiv.2204.02095}.
Although many streaming algorithms, including this one,
can be implemented in the MPC model,
they approximate the optimal value without reporting a solution;
see \Cref{sec:related} for details.
Thus, from a technical perspective,
our chief contribution is to compute an approximately optimal solution,
which is a set of facilities $F$, in the fully-scalable regime.

\paragraph{$k$-Clustering.}
Our second result presents MPC algorithms for \kMedian and \kMeans
that achieve an $(O(\mu^{-2}),1+\mu)$-bicriteria approximation,
for any desired $\mu\in(0,1)$.
As usual, \emph{$(\alpha,\beta)$-bicriteria approximation}
means that for every input, the algorithm outputs at most $\beta k$ centers
whose cost is at most $\alpha$-factor larger than the optimum cost for $k$ centers.
By letting $\mu>0$ be arbitrarily small but fixed,
our algorithm gets arbitrarily close (multiplicatively) to $k$ centers,
while still approximating the optimal cost within $O(1)$-factor;
in both respects, this is far stronger than the previous
bicriteria approximation for \kMeans~\cite{DBLP:conf/icml/BhaskaraW18}.
Our result is incomparable to the previous bound for \kMedian,
which is a true $O(\polylog(n))$-approximation \cite{DBLP:conf/nips/Cohen-AddadLNSS21}.
Nevertheless, our bicriteria approximation breaks
a fundamental technical barrier in their tree-embedding approach,
which cannot go below $O(\log n)$ ratio,
due to known distortion lower bounds,
and moreover does not generalize to \kMeans,
because it fails to preserve the squared distance.

Our approach is to tackle \kMedian and \kMeans via Lagrangian relaxation
and rely on our algorithm for \FL.
This approach can inherently handle large $k$,
because our core algorithms for \FL can even output $n$ clusters.
The Lagrangian technique was initiated by Jain and Vazirani~\cite{JV01},
who achieved a true $O(1)$-approximation for \kMedian
by leveraging special properties of their algorithm for \FL (see Section~\ref{sec:techniques}).
However, their primal-dual approach for \FL seems challenging to implement in fully-scalable MPC.
We thus develop an alternative approach that is inherently more parallel,
albeit achieves only bicriteria approximation for $k$-clustering.

\begin{theorem} [Simplified version; see \Cref{thm:clustering}]
\label{thm:clustering-intro}
Let $\eps,\sigma \in (0,1)$ be fixed.
There is a randomized fully-scalable MPC algorithm that,
given $\mu \in (0,1)$, $k \ge 1$, and
a multiset $P\subset\RR^d$ of $n$ points distributed across machines of memory $s \ge n^{\sigma}\cdot \poly(d)$,
computes in $O_\sigma(1)$ rounds an $(O_\eps(\mu^{-2}), 1+\mu)$-bicriteria approximation
for \kMedian, or alternatively \kMeans.
The algorithm uses $O(n^{1+\eps}) \cdot d^{O(1)}$ total space.
\end{theorem}

This simplified statement omits the same standard details as in \Cref{thm:ufl_intro},
and \Cref{remark:details} applies here as well.
In addition, it is known for \kMedian and \kMeans (but not known for \FL)
that any (true) finite approximation
requires $\Omega(\log_s n)$ rounds~\cite{DBLP:conf/icml/BhaskaraW18}.

\subsection{Technical Overview}
\label{sec:techniques}

We overview the main components in our algorithms,
and highlight technical ideas that may find more applications in the future.
We focus in this overview on \FL and \kMedian, noting that our algorithm
for \kMeans (and other powers $z\ge 1$) is essentially the same.
We further assume that $s = n^{\sigma}\cdot \poly(d)$ for a fixed $\sigma\in (0,1)$,
and we aim to achieve round complexity $O(\log_s n) = O_\sigma(1)$.

\paragraph{\FL.}
Several different algorithmic approaches have been used in the past
to achieve $O(1)$-approximation for \FL, including
LP-rounding~\cite{DBLP:journals/siamcomp/ByrkaA10,DBLP:journals/iandc/Li13},
primal-dual~\cite{JV01},
greedy~\cite{JMMSV03},
and local search~\cite{DBLP:journals/siamcomp/AryaGKMMP04}.
Some of these sequential algorithms were adapted to the PRAM model \cite{BT10},
i.e., to run in polylogarithmic parallel time (RNC algorithms).
These algorithms can be implemented in the MPC model,
but as some logarithmic factors in the time complexity seem inherent,
these fall short of $O(1)$ rounds in the fully-scalable regime. 
Our starting point is the Mettu-Plaxton (MP) algorithm~\cite{DBLP:journals/siamcomp/MettuP03},
which is a combinatorial algorithm inspired by~\cite{JV01},
that has been previously used to achieve $O(1)$-approximation in a few related models,
particularly streaming, congested clique, and sublinear-time computation~\cite{BCIS05,DBLP:journals/algorithmica/GehweilerLS14,arxiv.2204.02095}.
At a high level, this MP algorithm has two steps,
it first computes a ``radius'' $r_p > 0$ for every $p \in P$ (see \Cref{def:r_p}),
and then uses these $r_p$ values to determine which facilities to open.
However, implementing these steps in MPC is technically challenging:
\begin{itemize} [nosep]
\item
  Computing $r_p$ (approximately) can be reduced to
  counting the number of points in $P$ within a certain distance from $p$~\cite{BCIS05},
which is non-trivial to compute in $O(1)$ rounds,
  because many geometric techniques are ineffective in high dimension,
  for instance quadtrees/tree embeddings incur large approximation error
  and grids/nets require large memory.
  We overcome this issue by devising a new MPC primitive
  for geometric aggregation (in high dimension),
that can handle a wide range of statistics, including the counting mentioned above.
\item
  The MP algorithm determines which facilities to open
  by scanning the points in $P$ in order of non-decreasing $r_p$ value
  and deciding greedily whether to open each point as a facility.
  We design a new algorithm that avoids any sequential decision-making
  and decides whether to open each facility locally and in parallel.
  Our new algorithm may thus be useful also in other models.
\end{itemize}
Let us discuss these two new ideas in more detail.

\paragraph{MPC Primitive for Geometric Aggregation in High Dimension.}
We propose a new MPC primitive for aggregation tasks
in high-dimensional Euclidean spaces (see \Cref{thm:mpc} for details).
Given a radius $r > 0$, this primitive outputs, for every data point $p \in P$,
a certain statistic of the ball
$B_P(p, r) := P \cap \{ y \in \mathbb{R}^d : \dist(x, y) \leq r\}$.
Our implementation can handle any statistic that is defined by a \emph{composable} function $f$,
which means that $f(\cup_i S_i)$ for disjoint sets $S_i\subset \RR^d$
can be evaluated from the values $\set{f(S_i)}_i$.
Composable functions include counting the number of points,
or finding the smallest label (when data points are labeled).
In fact, this primitive can even be used
for approximate nearest-neighbor search in parallel for all points
(see \Cref{sec:nearest-neighbor-search}).
This natural aggregation tool plays a central role in all our MPC algorithms,
and we expect it to be useful for other MPC algorithms in high dimension.

Technically, exact computation of such statistics may be difficult,
and our algorithm estimates the statistic by evaluating it (exactly)
on a set $A_P(p, r)$ that approximates the ball $B_P(p, r)$
in the sense that it is sandwiched between $B_P(p, r)$ and $B_P(p, \beta r)$
for some error parameter $\beta \geq 1$.
To implement this algorithm in MPC, a natural idea is
to ``collect'' all data points in the ball $B_P(p, r)$ at a single machine,
however this set might be too large to fit in one or even few machines.
A standard technique to resolve this issue in low dimension
is to impose a grid of fine resolution, say $\epsilon r$,
and move each data point to its nearest grid point,
which provides a decent approximation (e.g., $\beta = 1+ \sqrt{d}\epsilon$).
However, a ball $B_P(p, r)$ might contain $\epsilon^{-\Theta(d)}$ grid points,
which for high dimension might still not fit in one machine.

Our approach is to use \emph{consistent hashing}
(see \Cref{def:geometric_hashing}),
which was first introduced in~\cite{DBLP:conf/stoc/JiaLNRS05}
under the name sparse partition,
and was recently employed in the streaming setting
for \FL in high dimension~\cite{arxiv.2204.02095}.
Roughly speaking, consistent hashing is a partition of the space $\RR^d$,
such that each part has diameter bounded by $\beta r$,
and every ball $B_P(p, r)$ intersects at most $n^{1 / \beta}$ parts.\footnote{This tradeoff between $\beta r$ and $n^{1 / \beta}$
  is just one specific choice of known parameters.
  Our theorem works with any possible parameters of consistent hashing,
  see \Cref{lemma:hash_bound} and \Cref{remark:jl}.
}
Our algorithm moves each data point $p\in P$
to a fixed representative point inside its own part,
then computes the desired statistic on each part
(namely, on the data points moved to the same representative),
and finally aggregates, for each $p\in P$,
the $n^{1 / \beta}$ statistics of the parts that intersect $B_P(p, r)$.
This algorithm can be implemented in $O(\log_{s} n)$ rounds on MPC,
albeit with slightly bigger total space $n^{1+ 1/\beta}$.
An interesting feature of this algorithm is that its core is deterministic
and hence leads to new \emph{deterministic} MPC algorithms,
including for approximate nearest-neighbor search;
see \Cref{sec:nns}.

\paragraph{Computing A Solution for \FL.}
Our algorithm is based on the MP algorithm,
  where a key notion is the ``radius'' $r_p>0$ defined for each data point $p\in P$.
  Formally, it takes the value $r$ such that
  serving all points in the ball $B(p, r)$ by $p$,
  i.e., opening a facility at $p$ and assigning points to $p$,
  incurs a cost of $|B_P(p, r)| \cdot r$; see \Cref{def:r_p}.
  This value $r$ always exists and is unique.
  It is known that a constant-factor approximation $\hat{r}_p$ of $r_p$
  satisfies that $|B_P(p, \hat{r}_p)| \approx 1 / \hat{r}_p$,
  hence computing $|B_P(p, r)|$ for $O(\poly\log n)$ different values of $r$
  suffices to compute an $O(1)$-approximation of $r_p$~\cite[Lemma 1]{BCIS05}.
  However, it is not easy to estimate $|B_P(p, r)|$ in MPC (simultaneously for all $p$ and $r$),
  and the abovementioned geometric aggregation only estimates $|A_P(p, r)|$
  for some $A_P(p, r)$ that is sandwiched between $B(p, r)$ and $B(p, \beta r)$.
  Nonetheless, we show this weaker estimate suffices for
  approximating $r_p$ within $O(\beta)$ factor (see \Cref{lemma:rp_offline}).
  Thus, our aggregation primitive yields an $O(1)$-approximation for all the $r_p$ values in $O(1)$ rounds.

An $O(1)$-factor estimate of the optimal cost $\OPT$ can be computed
from an $O(1)$-approximation of the $r_p$ values for all $p\in P$,
because $\sum_{p \in P} r_p = \Theta(\OPT)$ \cite{BCIS05,DBLP:journals/algorithmica/GehweilerLS14}.
Moreover, the $r_p$ values can be used to compute an $O(1)$-approximate \emph{solution} for \FL.
Specifically,
the MP algorithm~\cite{DBLP:journals/siamcomp/MettuP03}
scans the points in $P$ in order of non-decreasing $r_p$ value,
and opens a facility at each point $p$
if so far no facility was opened within distance $2r_p$ from $p$.
The sequential nature of this algorithm makes it inadequate for MPC,
and we therefore design a new algorithm that makes decisions in parallel.
It has two separate rules to decide whether to open a facility at each point $p\in P$:

\begin{enumerate}[nosep]
\item[(P1)]
  open a facility at $p$ with probability $\Theta(r_p)$, independently of other points; and
\item[(P2)]
  open a facility at $p$ if it has the smallest label among $B_P(p, r_p)$,
where each point $q\in P$ is assigned independently a random label $h(q)\in[0,1]$.
\end{enumerate}
Let us give some intuition for these rules.
Rule (P1) is a straightforward way to use the $r_p$ values
so that the expected opening cost is $O(\sum_p r_p) = O(\OPT)$,
but it is not sufficient by itself because the connection cost might be too large.
Indeed, if a cluster of points is very far from all other points,
say, a cluster of $t$ points all with the same $r_p=1/t$,
then with constant probability,
(P1) does not open any facility inside this cluster,
and the closest open facility is prohibitively far.
However, (P2) guarantees that at least one facility is opened inside this cluster,
at the point that has the smallest label.
Rule (P2) is not sufficient by itself as well.
Indeed, let $x_1\in P$ be the minimizer from the viewpoint of $p$,
i.e., have the smallest label in $B_P(p, r_p)$.
Then it may happen that there is no facility at $x_1$
because $B_P(x_1, r_{x_1})$ has another point $x_2$ with an even smaller label.
The same may happen also to $x_2$, and we may potentially get a long
``assignment'' sequence $(p=x_0, x_1, x_2, x_3, \ldots, x_t)$,
where each $x_{i+1}$ is the minimizer from the viewpoint of $x_i$ and only the last point $x_t$ is an open facility.
In this case, the connection cost of $p$ can be as large as
$\sum_{i=0}^{t-1} r_{x_i}$
(i.e., matching the bound obtained by the triangle inequality),
which might be unbounded relative to $r_p$.
This might happen not only for one point $p$ but actually for many points,
and the total connection cost would be prohibitive.

We deal with this assignment issue using rule (P1),
and in effect use both rules together.
Given such an assignment sequence, we prove that for every $i \ge 0$,
either (i) $B_P(x_i, r_{x_i})$ contains an open facility  with constant probability,
or (ii) $r_{x_{i+1}} \le r_{x_i}/2$, i.e., the $r_p$ value drops significantly in the next step.
Property (i) means that the sequence stops at $x_i$ with constant probability and thus, unless (ii) occurs,
the expected length of the assignment sequence is $O(1)$.
Property (ii) implies that the connection cost of the next step
drops significantly as it is proportional to its $r_p$ value,
and hence, we just sum up a subsequence of geometrically decreasing $r_p$ values.
Combining the two properties, we obtain that the expected connection cost for every point $p$ is $O(r_p)$.
Hence, the expected total connection cost is $O(\sum_p r_p) = O(\OPT)$.

The main technical challenge is to show property (i) when (ii) does not occur.
Specifically, we need to show that,
given a partial reassignment sequence $(p = x_0, \ldots, x_{i - 1})$,
the probability that the sequence does not terminate at $x_i$ is bounded by a small constant.
This event can be broken into two sub-events:
(a) rule (P1) does not open any facility at the ``new'' points of $B_P(x_i, r_{x_i})$,
where a point is considered new if it is not in $\cup_{j<i} B_P(x_j, r_{x_j})$;
and (b) the smallest label appears at a new point
(and thus rule (P2) does not open a facility at $x_i$).
Let $t\in [0,1]$ be the fraction of points that are new in $B_P(x_i, r_{x_i})$.
Then the probability of (a) is roughly
$(1 - r_{x_i})^{\Omega(t / r_{x_i})} \approx e^{-\Omega(t)}$,
where this calculation crucially uses that (ii) does not occur,
which means all new points have a similar $r_p$ value as $x_i$,
and that for every point $x$, the ball $B_P(x, r_x)$ roughly contains $\Omega(1 / r_x)$ points,
as formalized in \Cref{fact:rp}.
Since the two events are independent
and since (b) happens with probability $t$ by symmetry,
we conclude that the probability we need to bound is at most $\exp(-O(t)) \cdot t \leq O(1)$.
Here, one can observe that (P1) and (P2) are ``complementing'' each other to make the said probability small:
when $t$ is large, the probability $\exp(-O(t))$ of (a), which comes from rule (P1), is small;
otherwise, the probability $t$ of (b), which comes from rule (P2), is small.

The idea of opening facilities using random labels
and analyzing the cost by constructing an assignment sequence
was previously used in~\cite{DBLP:conf/soda/AnNS15}.
The context there is of a dynamically changing input,
and this technique is used to limit changes in the solution over time,
while our goal is to have a fast parallel implementation.
Although the high level idea is similar, the setup is quite different,
as their algorithm needs a solution to a linear-programming (LP) relaxation for \FL,
while ours only needs the $r_p$ values;
and consequently also the analysis is different,
as their analysis uses the LP constraints and bounds the cost relative to LP value,
while our analysis crucially uses basic properties of the $r_p$ values,
as in \Cref{fact:rp}.

It remains to bound the opening cost which is the easy part.
We show that the number of open facilities is $O(\sum_p r_p)$ in expectation.
Indeed, for points selected by rule (P1) this is clear.
For rule (P2), since $B_P(p, r_p)$ contains at least $1/r_p$ points,
the probability that $p$ has the smallest label in $B_P(p, r_p)$ is at most $r_p$.

When implementing our algorithm in the MPC model, the only non-trivial part (except for estimating the $r_p$ values)
is to check that $p$ has the minimum label in the ball $B(p,r_p)$.
Nevertheless, it is sufficient to look for the minimum label
in a larger set that approximates the ball,
such as a set sandwiched between this ball and a ball whose radius is larger by $O(1)$-factor.
Therefore, the MPC primitive for geometric aggregation is sufficient to execute rule (P2).

Overall, these rules compute a set of open facilities.
To compute also an assignment of each point $p\in P$ to an open facility,
we use our approximate nearest-neighbor search algorithm,
to find for each point $p\in P$ its $O(1)$-approximately closest facility
(see \Cref{sec:nearest-neighbor-search}).
Notice that the assignment algorithm searches for the closest facility,
but the analysis is still based on the assignment sequence as above,
even though the connection cost of this sequence may substantially exceed the distance to the closest facility.

\paragraph{$k$-\textsc{Clustering}.}
Our algorithm for \kMedian follows the Lagrangian-relaxation framework
established by Jain and Vazirani~\cite{JV01}
(and used implicitly earlier by Garg~\cite{DBLP:conf/focs/Garg96}).
They managed to obtain a true $O(1)$-approximation for \kMedian
by leveraging a special property guaranteed by their algorithm for \FL,
namely, its output solution has opening cost  $\cost_O$ and connection cost $\cost_C$
that satisfy $\alpha \cdot \cost_O + \cost_C \leq \alpha \cdot \OPT$
for a certain $\alpha=O(1)$.
Unfortunately, this stronger property is obtained via a highly sequential primal-dual approach, and seems difficult to implement efficiently in MPC,
particularly because it is too sensitive for the known toolkit for high dimension,
like locality sensitive hashing (LSH) and our geometric aggregation via consistent hashing.

We therefore take another approach of
relying on a \emph{generic} $\gamma$-approximation algorithm for \FL,
and using it in a black-box manner to obtain bicriteria approximation for \kMedian.
Our algorithm can output $(1+\mu)k$ centers whose cost is
at most $O(1/\mu^2)$-factor larger than the optimum cost for $k$ centers,
for any desired $0 < \mu < 1$.
This type of tradeoff, where the number of centers is arbitrarily close to $k$,
was relatively less understood,
as previous work has focused mostly on a smaller $O(1)$-factor in the cost,
but using significantly more than $k$ centers~\cite{DBLP:journals/ipl/LinV92,DBLP:conf/stoc/LinV92,DBLP:conf/approx/AggarwalDK09,DBLP:conf/nips/Wei16,DBLP:journals/corr/HsuT16}.
The result of~\cite{DBLP:conf/approx/MakarychevMSW16} does give $(1 + \mu)k$ centers,
and is thus the closest to ours in terms of bicriteria bounds,
however it relies explicitly on LP rounding, which seems difficult to implement in MPC.
To the best of our knowledge,
obtaining $(1+\mu)k$ centers for $k$-clustering by a black-box reduction to \FL
was not known before.
We believe this new reduction, and the smooth tradeoff it offers,
may be of independent interest.

In more detail, the black-box reduction from \kMedian to \FL goes as follows:
Assume momentarily that we know (an approximation of) the optimal clustering cost $\OPT$,
and consider the clustering instance as an input for \FL with opening cost $\wopen := \OPT / k$.
The optimal cost of this instance is at most $k\cdot \wopen + \OPT = 2\OPT$
as the optimal \kMedian solution is also a feasible solution for \FL with at most $k$ facilities;
note that the choice of $\wopen$ balances the opening and connection costs in this solution.
We then run (any) $\gamma$-approximation algorithm for \FL,
and it will find a solution whose cost is at most $\gamma\cdot 2\OPT$.
The choice of $\wopen$ implies that the number of open facilities in this solution
is at most $\gamma\cdot 2\OPT / \wopen = 2\gamma\cdot k$.
Hence, we obtain a $(2\gamma,2\gamma)$-bicriteria approximation of \kMedian.
Finally, we remove the assumption of knowing $\OPT$
by running this procedure in parallel
for a logarithmic number of guesses for $\OPT$,
and taking the cheapest solution that uses at most $O(\gamma\cdot k)$ centers.
Using our MPC algorithm for \FL, this gives an efficient MPC algorithm for clustering
with $(O(\gamma),O(\gamma))$-bicriteria approximation guarantees.

In this approach, the number of centers (in the solution)
might exceed $k$ by a large constant factor.
We now use a different method to reduce it to be arbitrarily close to $k$.
First, we observe that the $(O(\gamma), O(\gamma))$-approximate solution can be used to get a \emph{weak coreset},
namely, a weighted set of at most $O(\gamma\cdot k)$ distinct points,
such that any $\alpha$-approximate set of centers for the coreset
is also an $O(\alpha\cdot \gamma)$-approximate solution for the original instance.
To obtain the coreset, we just move every point to its approximately nearest facility in the approximate solution, which is a standard step in the clustering literature (see e.g.~\cite{DBLP:conf/focs/GuhaMMO00}).
Then, weight $w(p)$ of a point $p$ is the number of original points rounded to $p$.
We note that the problem becomes trivial if the coreset fits in one machine, and thus the interesting case is when $k$ is very large.

Given this weak coreset,
we find an approximate solution with at most, say, $2k$ centers
and cost increased by another factor of $O(\gamma)$,
using the following simple but sequential algorithm:
Process the coreset points in order of non-increasing weights $w(\cdot)$,
and open a center at point $p$
if so far there is no center within distance $\OPT / (k\cdot w(p))$ from $p$;
here we again need a guess for $\OPT$.
In a nutshell, the analysis of this algorithm is based on averaging arguments;
intuitively, only few points in the weak coreset can have a relatively large connection cost in the optimal solution.

We then convert this sequential algorithm into a parallel one
using similar ideas as for \FL.
That is, for every point in parallel, we add it to the set of centers
using two separate rules:
(i) independently with probability $1/\gamma$; or
(ii) if there is no point with larger weight within distance $\OPT / (k\cdot w(p))$.
To implement (ii) in MPC, we need to ensure a consistent tie-breaking,
for which a small random perturbation of the weights is sufficient.
Finally, to efficiently find the point of maximum (perturbed) weight
in the neighborhood of every point,
we again employ our MPC primitive for geometric aggregation.

 \subsection{Related Work}
\label{sec:related}

\paragraph{Parallel and Distributed Algorithms.}

A more general metric setting of \FL has been studied earlier in the distributed Congest and CongestedClique models
(see, e.g., \cite{DBLP:journals/algorithmica/GehweilerLS14,HP15,HPS14}), and these results immediately transfer
into MPC algorithms with $O(n)$ local memory and $O(n^2)$ total space.
In particular, in combination with the recent result in \cite{CKPU23}, these results yield
an $O(1)$-round MPC algorithm for metric \FL.
In this general-metric setting of \FL, instances have size $O(n^2)$,
which makes the problem significantly different from our geometric setting,
e.g., instances in $\RR^d$ are trivial if the local memory is $O(n)$.
Furthermore, those results for general metrics rely on computing $O(1)$-ruling sets,
and by the conditional lower bounds in \cite{GKU19,CDP21},
this seems to require $\omega(1)$ rounds on a fully-scalable MPC.
Our algorithms bypass the obstacle of ruling sets
by leveraging the geometric structure in $\RR^d$
and one can view our high-level contribution
as proposing a setting 
avoiding that obstacle.

Clustering problems (e.g., \FL, \kMeans, \kMedian) have been also studied in the PRAM model of parallel computation \cite{BT10,DBLP:conf/spaa/BlellochGT12}. These algorithms can be implemented in the MPC model, but the logarithmic factors in the running time or the approximation ratio seem inherent, and they fall short of achieving $O(1)$ rounds in the fully-scalable regime.

\paragraph{Connections to Streaming.}
A closely related model is the streaming model, which mainly focuses on sequential processing of large datasets by a single machine with a limited (sublinear) memory.
In general, if a streaming algorithm is storing only a linear sketch and uses space $O(s^{1-\varepsilon})$ for a fixed $\varepsilon > 0$,
then it can be simulated on MPC in $O(\log_s N)$ rounds
(recall $s$ is the local memory per machine and $N$ is the input size);
this was also observed and mentioned in, e.g.,~\cite{CCJLW23}.
Thus, the various recent results on streaming algorithms in high dimension
can be readily applied in MPC, and we briefly discuss the most relevant ones in the following.
We note that streaming algorithms typically assume the input is discrete, i.e.,
$P\subseteq [\Delta]^d$ with $\Delta = \poly(n)$.

For both \kMedian~\cite{DBLP:conf/icml/BravermanFLSY17} and \kMeans~\cite{DBLP:journals/corr/abs-1802-00459},
it is possible to find $(1 + \epsilon)$-approximate solution with $k$ centers using space $\poly(\epsilon^{-1}kd \log \Delta)$.
However, these algorithms are not directly applicable in our setting,
since to simulate these algorithms it requires local memory size $s = \Omega(k)$ which is not fully scalable.
For \FL, \cite{Indyk04} gave $O(d \log \Delta)$-approximation (along with several other problems including minimum spanning tree and matching), using space $\poly(d \log \Delta)$.
Later on, an $O(d / \log d)$-approximation for \FL was obtained using similar space $\poly(d \log \Delta)$,
and alternatively $O(1 / \epsilon)$-approximation using space $O(n^{\epsilon})$~\cite{arxiv.2204.02095}.
However, these results for \FL can only estimate the optimal cost,
as storing the solution requires linear space (which is too costly since streaming algorithms aim to use sublinear space).
Hence, simulating these results only leads to estimating the optimal cost in MPC,
while our result for \FL can indeed find an approximate solution.

\paragraph{MPC Algorithms for MST.}
A similar issue of cost estimation versus finding approximate solutions
is present also for the minimum spanning tree (MST) problem.
The classical streaming algorithm of~\cite{Indyk04} was recently improved to an $O(\log n)$-approximation using the same space regime $\poly(d \log \Delta)$~\cite{DBLP:conf/stoc/ChenJLW22}
and to an $O(1 / \epsilon^2)$-approximation for the $n^{\epsilon}$ space regime~\cite{CCJLW23}.
However, these results are for estimating the optimal MST cost, and it is still open to find an $O(1)$-approximate solution for high-dimensional MST in $O(\log_{s}n)$ rounds of MPC.
Indeed, the currently best MPC algorithm finds an $O(1)$-approximate MST in $\tilde{O}(\log \log n)$ rounds~\cite{JMNZ24} (using local space polynomial in $n$),
	while in $O(\log_{s} n)$ rounds, one can only find an $O(1)$-approximation in a low dimension~\cite{DBLP:conf/stoc/AndoniNOY14} (see also~\cite{CZ24} for an exact MST algorithm for $d=2$)
	or a $\poly(\log n)$-approximation in a high dimension~\cite{DBLP:conf/spaa/AhanchiAHKZ23}.
For a related problem of single-linkage clustering in a low dimension, there is a $(1+\varepsilon)$-approximation on MPC in $O(\log n)$ rounds~\cite{YaroslavtsevV18}.

 \section{Preliminaries}
\label{sec:prelim}

For integer $n$, let $[n] \eqdef \left\{1, 2, \ldots, n\right\}$.
For a function $\varphi: X\rightarrow Y$,
the image of a subset $S \subseteq X$ is defined 
$\varphi(S) \eqdef \{ \varphi(x) : x \in S\}$,
and the preimage of $y\in Y$ is defined as
$\varphi^{-1}(y) \eqdef \left\{x\in X : \varphi(x) = y\right\}$.
A Euclidean ball centered at $x\in \mathbb{R}^{d}$ with radius $r \geq 0$
is defined as $B(x,r) \eqdef \left\{y\in \mathbb{R}^{d} : \dist(x,y)\leq r\right\}$,
where $\dist(x, y) \eqdef \|x - y\|_2$ refers throughout to Euclidean distance. 
For a set $P\subset \RR^d$, we define $B_P(x, r)  \eqdef B(x, r) \cap P$,
which is also a metric ball inside $P$.
Let $\diam(S)$ denote the diameter of $S\subseteq\mathbb{R}^{d}$.  
The aspect ratio of a point set $S \subset \mathbb{R}^d$ is the ratio between the maximum and minimum inter-point distance of $S$, i.e.,
$\frac{\max_{x, y \in S} \dist(x, y)} {\min_{x \neq y \in S} \dist(x, y)}$.

\begin{fact}[Generalized triangle inequality]
  \label{lemma:tri_ine_z}
  Let $(V, \rho)$ be a metric space, and let $z \geq 1$. Then
  \[
    \forall x, x', y \in V,
    \qquad
    \rho^z(x, y) \leq 2^{z-1} (\rho^z(x, x') + \rho^z(x', y)) .
  \]
\end{fact}

\paragraph{\textsc{Power-$z$ (Uniform) Facility Location}.}
Given a set of data points $P \subset \mathbb{R}^d$, a (uniform) opening cost $\wopen > 0$ and some $z \geq 1$,
the objective function of  \textsc{Power-$z$ (Uniform) Facility Location}
for a set of facilities $F \subset \mathbb{R}^d$ is defined as 
\[
    \fl_z(P, F) \eqdef |F| \cdot \wopen + \sum_{p \in P} \dist^z(p, F),
\]
where again $\dist(x, y) \eqdef \|x - y\|_2$ and
$\dist(x, S) \eqdef \min_{y \in S} \dist(x, y)$. 
From now on, we omit ``uniform'' from the name of the problem, and simply use \pzFL.
We denote the minimum value of a solution for \pzFL by
$\OPTfl_z(P) \eqdef \min_{F\subseteq \mathbb{R}^{d}}\fl_z(P,F)$;
we omit $P$ if it is clear from the context.

\paragraph{\kzC.}
Given a set of data points $P \subset \mathbb{R}^d$, an integer $k \geq 1$ and some $z \geq 1$,
the objective function of \kzC
for a center set $C \subset \mathbb{R}^d$ with $|C| \leq k$ is defined as
\[
  \cl_z(P, C) \eqdef \sum_{p \in P} \dist^z(p, C).
\]
Notice that the special cases $z = 1$ and $z = 2$ are called \kMedian and \kMeans, respectively.
We denote the minimum value of a solution for \kzC by
$\OPTkz(P) := \min_{C\subseteq \mathbb{R}^{d}: |C|\le k}\cl_z(P, C)$;
we again omit $P$ when it is clear from the context.

\paragraph{Consistent Hashing.}
As mentioned above, our MPC primitive for geometric aggregation
relies on \emph{consistent hashing}.
We define it below, and then state the best known bounds for its parameters,
which are near-optimal~\cite{DBLP:conf/icalp/Filtser20}.

\begin{definition}[{\cite[Definition 1.6]{arxiv.2204.02095}}]
\label{def:geometric_hashing} 
A mapping $\varphi: \mathbb{R}^{d}\rightarrow\mathbb{R}^{d}$ is called a
\emph{$\Gamma$-gap $\Lambda$-consistent hash}
with diameter bound $\ell>0$, or simply $(\Gamma,\Lambda)$-hash,\footnote{Note that this definition is scale invariant with respect to $\ell$,
  i.e., a scaling of $\RR^d$ will scale $\ell$ but not affect
  the parameters $\Gamma$ and $\Lambda$.
  Thus, upper and lower bounds can restrict attention to the case $\ell=1$.
}
if it satisfies: 
\begin{enumerate} [nosep]
\item Diameter: for every image $z\in \varphi(\mathbb{R}^{d})$, we have $\diam(\varphi^{-1}(z))\leq \ell$; and
\item Consistency: for every $S\subseteq\mathbb{R}^{d}$ with $\diam(S)\leq \ell/\Gamma$, we have $|\varphi(S)|\leq \Lambda$.
\end{enumerate}
\end{definition}

\begin{lemma}[{\cite[Theorem 5.1]{arxiv.2204.02095}}]
  \label{lemma:hash_bound}
  For every $\Gamma \in [8, 2d]$, 
  there exists a (deterministic) $(\Gamma, \Lambda)$-hash $\varphi:\RR^d\to\RR^d$
  where $\Lambda = \exp(8d / \Gamma) \cdot O(d \log d)$.
  Furthermore, $\varphi$ can be described using $O(d^2 \log^2 d)$ bits
  and one can evaluate $\varphi(x)$ for any point $x \in \mathbb{R}^d$ in space $O(d^2 \log^2 d)$.
\end{lemma}

\begin{remark}
\label{remark:jl}

Our main results also hold under the assumption that $s \geq \polylog(dn)$,
which is preferable when $d\gg \log n$.
It follows by the well-known idea of applying
a JL transform $\pi$ (named after Johnson and Lindenstrauss~\cite{JL84})
with target dimension $O(\log n)$
as a preprocessing of the input $P \subset \RR^d$,
i.e., running our algorithm on $\pi(P)$ instead of $P$.
Then, with probability $1 - 1 / \poly(n)$,
all the guarantees in our results would suffer only an $O(1)$ factor. 
To implement this preprocessing in MPC using $\polylog (dn)$ words of local memory,
we use a bounded-space version of the JL transform~\cite{PRG00},
that requires only $\polylog (dn)$ words to specify $\pi$
(in comparison, a naive implementation of $\pi$ requires $O(d \log n)$ words).
One machine can randomly generate this specification of $\pi$ and broadcast it,
and then all machines can apply $\pi$ locally in parallel. 
This requires additional $O(\polylog (dn))$ local memory and $O(\log_{s}n)$ rounds,
and these additional costs are easily absorbed in our bounds.
Hence, in all our results we can assume without loss of generality
that $P$ is replaced by $\pi(P)$. 
Furthermore, after this preprocessing,
we can use $d = O(\log n)$ also in the consistent hashing bounds in \Cref{lemma:hash_bound}
to obtain a tradeoff $\Gamma = O(1 / \eps)$ and $\Lambda = O(n^{\eps})$,
for any desired fixed $\eps\in(0,1)$, 
and also the hash can be described using $\polylog (n)$ bits.
\end{remark}

 \section{MPC Primitive for Geometric Aggregation in High Dimension}
\label{sec:mpc}

In our MPC algorithms, we often face a scenario
where we want to compute something for each input point $p\in P$.
That computation is a relatively simple problem,
like computing the number of points in the ball $B_P(p,r)$ for some global value $r>0$.
A more general version is to allow different radii $r$ (local for each $p$);
another generalization is to compute some function $f$ over the points in $B_P(p,r)$,
like finding the point with smallest identifier or smallest distance to $p$.
A naive approach is to collect (copies of) all the points in $B_P(p,r)$
to the same machine, say the one holding $p$, and then compute $f$ there.
This is very challenging and our solution is to approximate these balls
by generating sets $A_P(p,r) \approx B_P(p,r)$,
and evaluate $f$ on these sets instead of on the balls.
The approximation here just means that the set $A_P(p,r)$
is sandwiched between a ball of radius $r$ and one of larger radius,
see~\eqref{eq:MPCsandwich}.
Informally, we thus compute $f(A_P(p,r)) \approx f(B_P(p,r))$,
but of course the approximation here need not be multiplicative.

Our MPC algorithm derives these sets $A_P(p,r)$ from consistent hashing (see Definition~\ref{def:geometric_hashing}),
and thus our description below requires access to such a hash function $\varphi$,
and moreover the final guarantees depend on the parameters $\Gamma$ and $\Lambda$ of the consistent hashing.
The theorem below is stated in general, i.e., for any $(\Gamma, \Lambda)$-hash,
but we eventually employ known constructions 
with $\Gamma = O(1/\eps)$ and $\Lambda = n^\eps$, for any desired $\eps>0$ (see \Cref{remark:jl} and \Cref{lemma:hash_bound} for details).  
Obviously, running this algorithm in MPC requires
an implementation of consistent hashing, which might require additional memory;
but this memory requirement is typically much smaller than $\sqrt{s}$,
and thus the hash function can be easily stored in each machine.

Yet another challenge is
that the entire set $A_P(p,r)$ might not fit in a single machine,
and we thus impose on $f$ another requirement.
We say that a function $f$ is \emph{composable}
if for every disjoint $S_1, \ldots,S_t \subseteq \RR^d$,
one can evaluate $f(S_1\cup\cdots\cup S_t)$
from the values of $f(S_{1}), \ldots, f(S_{t})$.\footnote{We use general $t$ here because of our intended application,
  but obviously it follows from the special case $t=2$.
}
In our context, $f$ maps finite subsets of $\RR^d$ to $\RR$.  For instance, $f(S)=|S|$ is clearly composable. 
For a few more interesting examples,
suppose every $x\in\RR^d$ is associated with a value $h(x)\in\RR$.
Now if $h(x)$ represents the weight of $x$,
then $f(S) = \sum_{x\in S} h(x)$ is the total weight of $S$;
and if $h(x)$ represents an identifier (perhaps chosen at random),
then $f(S) = \min_{x\in S} h(x)$ is the smallest identifier in $S$.

\begin{theorem} [Geometric Aggregation in MPC]
\label{thm:mpc}
There is a deterministic fully-scalable MPC algorithm with the following guarantees. 
Suppose that
\begin{itemize} \item the input is $r \geq 0$
     and a multiset $P\subset \RR^d$ of $n$ points
     distributed across machines with local memory $s\geq \poly(d\log n)$; and
   \item the algorithm has access
to a composable function $f$ (mapping finite subsets of $\RR^d$ to $\RR$) and
to a $(\Gamma, \Lambda)$-hash $\varphi:\RR^d\to\RR^d$.
\end{itemize}
Then the algorithm uses $O(\log_{s}n)$ rounds
and $O(\Lambda\cdot \poly(d)) \cdot \Ot(n)$ total space,
and computes for each $p\in P$ a value $f(A_P(p, r))$,
where $A_P(p, r)$ is an arbitrary set that satisfies
\begin{equation} \label{eq:MPCsandwich}
  B_P(p, r)\subseteq A_P(p, r)\subseteq B_P(p, 3\Gamma\cdot r).
\end{equation}
(In fact, the set $A_P(p, r)$ is determined by $\varphi$.)
\end{theorem}

\begin{proof}
Our algorithm makes use of the following standard subroutines in MPC, and we note that they are deterministic.
In the broadcast procedure, to send a message of length at most $\sqrt{s}$ from some machine $\mathcal{M}_0$ to every other ones,
one can build an $\sqrt{s}$-ary broadcasting tree whose nodes are the machines (with $\mathcal{M}_0$ as the root),
and send/replicate the message level-by-level through the tree (starting from the root).
Observe that the height of the tree is $O(\log_{s}n)$ and hence the entire process runs in $O(\log_{s}n)$ rounds.
The reversed procedure defined on the same broadcast tree, sometimes called converge-cast~\cite{mpc_book},
can be used to aggregate messages with size at most $\sqrt{s}$ distributed across machines to some root machine $\mathcal{M}_0$ (for instance, to aggregate the sum of vectors of length $\sqrt{s}$ distributed across machines), in $O(\log_{s}n)$ rounds.
In particular, it can be used to evaluate the composable function $f$ on a (distributed) set $S$,
where each machine $\mathcal{M}$ evaluates $f(S_{\mathcal{M}})$ for its own part $S_{\mathcal{M}} \subseteq S$,
and aggregate using converge-cast.

   \begin{lemma}[{Sorting in MPC~\cite{DBLP:conf/isaac/GoodrichSZ11}}]
      \label{lemma:sorting}
      There is a deterministic MPC algorithm that given a set $X$ of $N$ comparable items distributed across machines with local memory $s$,
      sorts $X$ such that each $x \in X$ knows its rank and $\forall x < y \in X$, it holds that the machine that holds $x$ has an ID no larger than that of $y$.
The algorithm uses $O(\log_{s}N)$ rounds and total space of $O(N)$ words. 
   \end{lemma}

We give an outline of our algorithm in \Cref{alg:mpc_compute_f};
the implementation details in MPC are discussed below.
The algorithm starts with ``partitioning'' $\mathbb{R}^d$ into buckets with respect to the $(\Gamma, \Lambda)$-hash $\varphi$,
and approximates each ball $B(p, r)$ by the union of buckets that this ball intersects.
This distorts the radius by at most an $O(\Gamma)$-factor.
Then, we evaluate the $f$ value on each bucket,
and the approximation to $f(B_P(p, r))$ is obtained by ``aggregating''
the $f$ value for the intersecting buckets of $B(p, r)$, where the composability of $f$ is crucially used.

\paragraph{Implementation Details.}
Here we discuss how each step of \Cref{alg:mpc_compute_f} is implemented efficiently in MPC.
In line~\ref{line:inverse}, since after the sorting, for each $u \in \varphi(P)$ the points in $P_u$,
i.e., the set of points $p$ such that $u = \varphi(p)$,  span a (partial) segment of machines with contiguous IDs,
one can use a converge-cast in parallel in each segment to aggregate $f(P_u)$.
In line~\ref{line:initial}, although the total space is sufficient to hold all tuples,
we may not have enough space to store the $O(\Lambda)$ tuples for a point $p$ in a single machine.
Instead, we allocate for every point $p$ a (partial) segment of machines whose total space is $O(\Lambda)$ (which can be figured out via sorting),
replicate $p$'s to them (via broadcast),
and generate $\varphi(B(p, r))$ in parallel on those machines.
Specifically, each machine in the segment is responsible for generating a part of $\varphi(B(p, r))$,
and a part can be generated locally without further communication since every machine holds the same deterministic $\varphi$.
Line~\ref{line:fPu} and line~\ref{line:fAp} can be implemented similarly by broadcast and converge-cast, respectively, in parallel on each segment of machines.

\begin{algorithm}[ht]
   \caption{MPC algorithm for evaluating $f(A_P(p, r))$ for $p \in P$, for given $P \subset \mathbb{R}^d, r > 0$}
   \label{alg:mpc_compute_f}
    \begin{algorithmic}[1]

\State each machine imposes the same $(\Gamma, \Lambda)$-hash $\varphi : \mathbb{R}^d \to \mathbb{R}^d$ with diameter bound $\ell := 2\Gamma r$ 

        \Comment{notice that $\varphi$ is deterministic, hence no communication is required}

        \State sort $P$ with respect to $\varphi(p)$ for $p \in P$ (using \Cref{lemma:sorting})

        \State for $u \in \varphi(P)$, evaluate and store $f(P_u)$ where $P_u := \varphi^{-1}(u) \cap P$ 
        \label{line:inverse}

        \State for each $p \in P$ and $u \in \varphi(B(p, r))$,
        create and store a tuple $(p, u)$
        \label{line:initial}

        \Comment{as $|\varphi(B(p, r))|=O(\Lambda)$ by \Cref{def:geometric_hashing}, the total space is enough to hold all tuples}

        \State sort the tuples with respect to $u$ (using \Cref{lemma:sorting})

        \State let $\mathcal{T}_u = \{(\cdot, u)\}$,
        append $f(P_u)$ to all tuples in $\mathcal{T}_u$
\label{line:fPu}

        \Comment{$f(P_u)$ is already evaluated and stored, as in line~\ref{line:inverse}}
\State sort the tuples with respect to $p$, and evaluate $f( A_P(p, r) )$ for each $p$, where
        \[
            A_P(p, r) := \bigcup_{u \in \varphi(B(p, r))} P_u 
         \]
         \label{line:fAp}
\end{algorithmic}
\end{algorithm}

\paragraph{Round Complexity and Total Space.}
The round complexity is dominated by the sorting, broadcast and converge-cast procedures,
which all take $O(\log_{s}n)$ rounds to finish and are invoked $O(1)$ times in total.
Therefore, the algorithm runs in $O(\log_{s}n)$ rounds.
The total space is asymptotically dominated by $\poly(d \log n)$ times the total number of tuples, which is $O(\Lambda)\cdot n$ by Definition~\ref{def:geometric_hashing}.

\paragraph{Correctness.}
Observe that the algorithm is deterministic, and hence there is no failure probability.
It remains to show that $A_P(p, r) $ satisfies that $B_P(p, r) \subseteq A_P(p, r) \subseteq B_P(p, 3\Gamma r)$.
Recall that $P_u = \varphi^{-1}(u)\cap P$ as defined in line~\ref{line:inverse},
and that $A_P(p, r)  = \bigcup_{u \in \varphi(B(p, r))} P_u$ as in line~\ref{line:fAp}.
Hence, we have
\[
   A_P(p, r) = P \cap \varphi^{-1}(\varphi(B(p, r))).
\]
Therefore, the first inequality is straightforward as $B_P(p, r)\subseteq \varphi^{-1}(\varphi(B(p, r)))\cap P=A_P(p, r)$ for any mapping $\varphi$.   

To prove the second inequality, fix some $p \in P$.
For a point $q\in A_P(p, r)$, by definition there is a $u_{q}\in \varphi(B_P(p, r))$ such that $q\in \varphi^{-1}(u_{q})$.
      Then by \Cref{def:geometric_hashing}, we have $\diam(\varphi^{-1}(u_{q}))\leq \ell= 2\Gamma r$
      which implies that any point $x\in \varphi^{-1}(u_{q})$ satisfies $\dist(x, q)\leq 2\Gamma r$. 
      Now, pick a point $x \in B_P(p, r)$ such that $\varphi(x) = u_q$; such a point exists as $u_{q}\in \varphi(B_P(p, r))$.
      Then by definition, $x \in B_P(p, r) \cap \varphi^{-1}(u_q)$ and $\dist(p, x)\leq r$. 
      Hence, by triangle inequality,
      we have that $\dist(p, q)\leq \dist(p, x) + \dist(x, q)\leq r + 2\Gamma r 
      \leq 3 \Gamma r$,
      which implies that $A_P(p, r)\subseteq B_P(p, 3\Gamma r)$. 
      This finishes the proof.
   \end{proof}

\subsection{Application to Nearest Neighbor Search}\label{sec:nearest-neighbor-search}
\label{sec:nns}
Given a set $X \subseteq \mathbb{R}^d$ of \emph{terminals} and a set $P \subseteq \mathbb{R}^d$ of data points,
the $\rho$-approximate nearest neighbor search problem asks to find for every $p \in P$ a terminal $x \in X$,
such that $\dist(p, x) \leq \rho \cdot \dist(p, X)$.
This process is useful in clustering and facility location problems,
since one can find an assignment of every data point to its approximately  nearest center/facility.
We show how to solve this problem using \Cref{thm:mpc},
provided the knowledge of the \emph{aspect ratio} $\Delta$ of $X\cup P$.

Pick an arbitrary point $x\in X\cup P$,
compute in $O(\log_s n)$ rounds 
the maximum distance $M := \max_{y\in X\cup P} \dist(x,y)$
from $x$ to every other point (via broadcast and converge-cast). 
Since $M$ is a $2$-approximation to $\diam(X \cup P)$,
we conclude that for every $x \neq y \in X \cup P$,
$M / \Delta \leq \dist(x, y) \leq 2M$.
Rescale the instance by dividing $M / \Delta$, then the distances are between $1$ and $O(\Delta)$.
Then, let $Z := \{ 2^i : 1 \leq 2^i \leq O(\Delta) \} $.
We apply \Cref{thm:mpc} in parallel for $r \in Z$ and $f$ such that $f(Y)$ for $Y \subseteq X$ finds the terminal with the smallest ID in $Y$ (where the ID of a point can be defined arbitrarily as long as it is consistent),
and $f$ returns $\bot$ if $Y = \emptyset$. This $f$ is clearly composable.
After we obtain the result of \Cref{thm:mpc}, i.e., $f(A_X(p, r))$ for $p \in P$ and $r \in Z$,
we find in parallel for each $p \in P$ the smallest $r \in Z$ such that $f(A_X(p, r)) \neq \bot$.
This way, we explicitly get for each point $p$ an approximately nearest facility in $X$.
This algorithm has approximation factor $\rho = O(\Gamma)$,
using total space by an $O(\log \Delta)$-factor larger than that of \Cref{thm:mpc},
while the round complexity remains $O(\log_{s}n)$.

We remark that techniques based on locality sensitive hashing (LSH)
can also be applied in MPC to solve the approximate nearest neighbor problem~\cite{DBLP:conf/icml/BhaskaraW18, CMZ22}.
LSH in fact achieves a slightly better tradeoff,
namely, an $O(c)$-approximation using total space $n^{1/c^2}\cdot \Ot(n)$,
while our approach requires total space $n^{1/c}\cdot \Ot(n)$,
by plugging in \Cref{lemma:hash_bound}
and assuming the dimension reduction in \Cref{remark:jl} is performed. 
Alternatively, if $d \leq O(\log n)$, one can obtain the same $n^{1/ c} \cdot \Ot(n)$ bound without applying the randomized dimension reduction in \Cref{remark:jl},
which leads to a deterministic algorithm,
whereas approaches based on LSH are inherently randomized.

 \section{Power-$z$ Facility Location}

In this section we design algorithms for Facility Location
and prove the following generalization of \Cref{thm:ufl_intro}.
Recall that our algorithms compute a solution
(e.g., a set of facilities) 
and not just its value. 

\begin{theorem}
\label{thm:ufl}
There is a randomized fully-scalable MPC algorithm
for approximating \pzFL, where $z\ge 1$ is a parameter,
with the following guarantees.
Suppose that
\begin{itemize} \item the input is an opening cost $\wopen>0$ and 
	a multiset $P$ of $n$ points in $\RR^d$ distributed across machines with local memory $s \geq \poly(d \log n)$; and
\item the algorithm has access to the same $(\Gamma, \Lambda)$-hash $\varphi:\RR^d\to\RR^d$ in every machine.
\end{itemize}
The algorithm uses $O(\log_{s}n)$ rounds and $O(\Lambda\poly( d)) \cdot \tilde{O}(n)$ total space,
and succeeds with probability at least $1-1/\poly(n)$
to compute a $\Gamma^{O(z)}$-approximate solution.
\end{theorem}

At a high level, our algorithm for \Cref{thm:ufl} is based on the Mettu-Plaxton algorithm~\cite{DBLP:journals/siamcomp/MettuP03}, particularly a variant in~\cite{DBLP:journals/algorithmica/GehweilerLS14} modified to handle the power-$z$ case.
We start with defining a key notion of $r_p$ value (per data point $p \in P$),
which was first introduced in~\cite{DBLP:journals/siamcomp/MettuP03}
and later generalized to the power-$z$ case in~\cite{DBLP:journals/algorithmica/GehweilerLS14}.

\begin{definition}[\cite{DBLP:journals/siamcomp/MettuP03,DBLP:journals/algorithmica/GehweilerLS14}] \label{def:r_p}
   For every $p \in P$, let $r_p > 0$ be such that
    \begin{equation} \label{eq:rp}
        \sum_{x \in B_P(p, r_p)} \big( r_p^z - \dist^z(p, x) \big) = \wopen.
    \end{equation}
\end{definition}

It is easy to see that $r_{p}$ is well-defined
(because the LHS is continuous and increasing with respect to $r_p\ge0$),
and that
\begin{equation}
   \label{eqn:rp_range}
   \frac{\wopen}{|P|}\leq r_{p}^z \leq \wopen. 
\end{equation}
Next, we state several important properties of the $r_p$ values,
particularly that the sum of $r_p^z$ is a good approximation for $\OPTfl_z$ (\Cref{fact:rp}),
and that these $r_p$ values are ``smooth'',
i.e., nearby points have comparable values (\Cref{lemma:rp_ub});
for $z=1$ this is actually a Lipschitz condition.

\begin{fact}[Lemmas 1 and 2 in \cite{DBLP:journals/algorithmica/GehweilerLS14}]
\label{fact:rp}
The following holds:
\begin{itemize}
\item For every $p \in P$,
  $|B_P(p, r_p)|\ge \wopen / r_p^z$.
\item $2^{-O(z)} \OPTfl_z \leq \sum_{p \in P}{r_p^z} \leq 2^{O(z)}  \OPTfl_z$.
\end{itemize}
\end{fact}

\begin{claim}
\label{lemma:rp_ub}
For all $p, q \in P$ and $z \geq 1$, it holds that
$r_q \leq  2^{(z - 1) / z} (r_p + \dist(p, q))$.
\end{claim}
\begin{proof}
By the triangle inequality, a ball around $p$ of radius $r_p$
is contained in a ball around $q$ of radius $\dist(q,p)+r_p$, 
i.e., $B(p, r_p) \subseteq B(q, r_p+\dist(q,p))$.
Thus, if we consider the value
$2^{(z - 1) / z}(r_p + \dist(p, q)) \geq r_p + \dist(p, q)$ as a ``candidate'' for $r_q$
and plug it into the LHS of~\eqref{eq:rp}, we see that
\begin{align*}
  \sum_{x \in B_P(q, 2^{(z - 1) / z}(r_p + \dist(q,p)))} & \left[ 2^{z - 1}( r_p + \dist(q,p))^z - \dist^z(q,x) \right] \\
  &\geq \sum_{x \in B_P(q, r_p + \dist(q,p))} \left[ 2^{z - 1} r_p^z + 2^{z - 1}\dist^z(q,p) - \dist^z(q,x) \right] \\
  &\geq \sum_{x \in B_P(q, r_p + \dist(q,p))} \left[ 2^{z - 1} r_p^z - 2^{z - 1} \dist^z(p, x) \right] \\
&\geq 2^{z - 1} \sum_{x \in B_P(p, r_p)} \left[ r_p^z - \dist^z(p,x) \right]
  = 2^{z - 1}\wopen \geq \wopen. 
\end{align*}
Since the LHS of~\eqref{eq:rp} is increasing,
we conclude that $r_q \leq 2^{(z - 1) / z} (r_p+\dist(q,p))$. 
\end{proof}

\paragraph{Proof Plan for \Cref{thm:ufl}.}
Our algorithm first estimates $r_p$ for each $p \in P$ and
then finds a set of facilities whose expected cost
is bounded by $O(1) \sum_p r^z_p$.
The MPC implementations of both steps rely on the geometric primitive
provided by \Cref{thm:mpc}.
We first describe these two steps as ordinary (not MPC) algorithms, 
in \Cref{sec:estimate_rp,sec:finding_fac_offline};
importantly, these algorithms require access to balls of the form $B(p,r)$,
but they work well with the ``approximation'' provided by \Cref{thm:mpc}.
We then convert these algorithms to MPC algorithms in \Cref{sec:mpc_ufl},
which completes the proof of \Cref{thm:ufl}.

\subsection{Estimating $r_p$ Values}
\label{sec:estimate_rp}

The next lemma provides a straightforward way to estimate $r_p$ (for every $p\in P$),
given the cardinalities of subsets $A_P(p, r) \subseteq P$
that ``approximate'' $B_P(p, r)$ in the sense of \Cref{thm:mpc}.
The key property used in the proof is that the LHS of \eqref{eq:rp} is increasing with respect to $r_p$.
Notice that our fast running time relies on
using the approximate sets $A_P(p, r)$ instead of the exact sets $B_P(p, r)$,
which might be harder to compute (require larger running time) in the MPC model.

\begin{lemma}[Approximating $r_p$ value]
\label{lemma:rp_offline}
Let $\beta > 1$ and suppose that for every $p \in P$ and $r \geq 0$
one has access to a subset $A_P^\beta(p, r)\subseteq P$ satisfying
$B_P(p, r) \subseteq A_P^\beta(p, r) \subseteq B_P(p, \beta r)$.
Given $p\in P$, let $\hat{r}>0$ be the smallest integral power of $\beta$ such that $|A_P^\beta(p, \hat{r})| \geq \wopen / (2 \beta^z \hat{r}^z)$.
Then
\[
  r_p / (3\beta) \leq \hat{r} \leq \beta r_p.
\]
\end{lemma}

\begin{proof}
If $r \geq r_p$, then
\begin{equation}
   \label{eqn:r_geq_rp}
   r^z |A_P^\beta(p, r)| 
   = \sum_{x \in A_P^\beta(p, r)} r^z
   \geq \sum_{x \in A_P^\beta(p, r)}   \big( r_p^z - \dist^z(p, x) \big)
   \geq \sum_{x \in B_P(p, r_p)} \big( r_p^z - \dist^z(p, x) \big)
   = \wopen.
\end{equation}
And if $r < r_p / (3\beta)$, then
\begin{align}
   2\beta^z r^z |A_P^\beta(p, r)| 
   &\leq \sum_{x \in A_P^\beta(p, r)} (3\beta r)^z - (\beta r)^z
   \leq \sum_{x \in A_P^\beta(p, r)} (3 \beta r)^z - \dist^z(p, x) \nonumber \\
   &< \sum_{x \in A_P^\beta(p, r)} r_p^z - \dist^z(p, x)
   \leq \sum_{x \in B_P(p, \beta r)} r_p^z - \dist^z(p, x) \nonumber \\
   &< \sum_{x \in B_P(p, r_p)} r_p^z - \dist^z(p, x) = \wopen. \label{eqn:r_l_rp}
\end{align}
The lower bound $\hat{r} \geq r_p / (3\beta)$
follows from the contrapositive of \eqref{eqn:r_l_rp}. 
For the upper bound, by definition of $\hat{r}$ we know that
$|A_P^\beta(p, \hat{r} / \beta)| < \wopen / (2 \beta^z (\hat{r}/\beta)^z) < \wopen / \hat{r}^z$.
Now if we assume that $\hat{r} > \beta r_p$, 
then we get from \eqref{eqn:r_geq_rp} that 
$|A_P^\beta(p, \hat{r} / \beta)| \geq \wopen / (\hat{r}/\beta)^z > \wopen / \hat{r}^z$,
arriving at contradiction.
Therefore, $r_p / (3\beta) \leq \hat{r} \leq \beta r_p$.
\end{proof}

\subsection{Opening the Facilities}
\label{sec:finding_fac_offline}

We proceed to describe our algorithm for computing the facilities,
i.e., selecting points from $P$ to open as facilities,
and we rely henceforth on a few easily justified assumptions.
First, we assume that we already have for every $p \in P$ 
an estimate $\hat{r}_p\in [r_p,\alpha r_p]$ for fixed $\alpha \geq 1$;
this is justified by applying \Cref{lemma:rp_offline} to obtain these estimates,
eventually setting $\alpha \eqdef 3\beta^2$ (in \Cref{sec:mpc_ufl}). 
Second, by rounding these estimates to the next power of $2$,
we can assume that whenever $\hat{r}_p > \hat{r}_q$, 
then actually $\hat{r}_p \geq 2 \hat{r}_q$
(recall that we aim for $O(1)$-approximation and do not optimize the constant).
Third, we assume that for a parameter $\beta \geq 1$ and every $p\in P$,
one has access to a subset $A^\beta_P(p, \hat{r}_p) \subseteq P$ satisfying 
$ B_P(p, \hat{r}_p) \subseteq A^\beta_P(p, \hat{r}_p) \subseteq B_P(p, \beta \hat{r}_p)$;
this is justified as before, by applying \Cref{thm:mpc}.
We present our algorithm for computing the set of facilities $F\subseteq P$
in \Cref{alg:offline_choose_fac},
and then prove that it achieves $(\alpha \beta)^{O(z)}$-approximation in expectation.
More specifically, the opening cost
is bounded in \Cref{lemma:offline_fac_opening}, which is rather immediate, 
and the connection cost is bounded in \Cref{lemma:offline_fac_finding},
whose proof is more lengthy and involved.

\begin{algorithm}[H]
   \caption{Computing a facility set $F\subseteq P$, given $P$, $\hat{r}_p \in [r_p, \alpha r_p]$ and $A_P^\beta(p, \hat{r}_p)$ for every $p \in P$}
   \label{alg:offline_choose_fac}
    \begin{algorithmic}[1]

      \State let $F \gets \emptyset$, and 
      let $\tau  \gets  (\alpha \beta)^{\Theta(z)} $ 

      \State for each $p \in P$, pick a uniformly random \emph{label} $h(p) \in [0, 1]$ 
      \State for each $p \in P$: 
         \begin{itemize}
            \item[(P1)] add $p$ to $F$ with probability $\tau \cdot \hat{r}_p^z / \wopen$ (independently of $h$)
            \item[(P2)] add $p$ to $F$ if $p$ has the smallest label in $A^\beta_P(p, \hat{r}_p)$
         \end{itemize} \label{line:open}
\end{algorithmic}
\end{algorithm}

\begin{lemma}[Opening Cost]
   \label{lemma:offline_fac_opening}
   Let $F \subseteq P$ be the set returned by \Cref{alg:offline_choose_fac}.
   Then $\E[|F| \cdot \wopen] = \alpha^{O(z)} \tau \cdot \OPTfl_z$.
\end{lemma}

\begin{proof}
Each point $p \in P$ is opened only if it is selected in either (P1) or (P2) in line~\ref{line:open} of \Cref{alg:offline_choose_fac},
hence its expected opening cost is bounded by the sum of the two cases.
In case (P1), $p$ is selected with probability $\tau \cdot \hat{r}_p^z / \wopen$,
hence its expected opening cost is $\tau \cdot \hat{r}_p^z$.
In case (P2), $p$ is selected with probability 
$
  \frac{1}{ |A^{\beta}_P(p, \hat{r}_{p})| }
  \leq \frac{1}{ |B_P(p, \hat{r}_p)| }
  \leq \frac{1}{ |B_P(p, r_p)| }
  \leq \frac{r_p^z}{\wopen} 
$,
where the last inequality is by \Cref{fact:rp},
hence its expected opening cost is at most $r_p^z$. 
The expectation of the total opening cost is thus at most
$(\tau+1) \sum_{p} \hat{r}_p^z \leq  \tau \alpha^z 2^{O(z)}\OPTfl_z$,
where we used \Cref{fact:rp}.
\end{proof}

\begin{lemma}[Connection Cost]
   \label{lemma:offline_fac_finding}
   Let $F \subseteq P$ be the set returned by \Cref{alg:offline_choose_fac}.
   Then
   \[
      \E[\sum_{p \in P} \dist^z(p, F)] = 2^{O(z)}\alpha^z \beta^z \cdot \OPTfl_z.
   \]
\end{lemma}

\begin{proof}
It suffices to bound the expected connection cost of each $p\in P$ by
$\E[\dist^z(p, F)] = 2^{O(z)}\alpha^z \beta^z \cdot r_p^z$, 
because we have $\sum_{p} r_p^z \le 2^{O(z)}\OPTfl_z$ by \Cref{fact:rp}.
To this end, fix henceforth $p\in P$
and consider $\dist(p, F)$, which is not easy to analyze directly,
because it depends on the closest point to $p$ that is opened as a facility.

\paragraph{Bounding the Connection Cost via a Reassignment Process.}
We analyze the distance $\dist(p, F)$
by identifying a sequence of points $S = (x_0,\ldots,x_t)$
that starts at $x_0:=p$ and its length $t\ge0$ is random (not predetermined).
We will bound $\dist(p, F)$ using the ``weight'' of this sequence,
defined as the total distance traveled by following this sequence of points,
taking care of the power $z$ using the generalized triangle inequality,
which asserts that $(a + b)^z  \leq 2^{z - 1} (a^z + b^z)$ for all $a, b\geq 0$.

We construct this sequence using an auxiliary algorithm,
described formally in \Cref{alg:def_seq}.
The idea is to test whether the current point $p$ (or a nearby point) is open by (P1) or (P2),
and if not then greedily move to a nearby point;
repeating this step until reaching an open facility forms a sequence $S$,
which we view as a process that repeatedly reassigns the current point.
We emphasize that \Cref{alg:def_seq} and the sequence $S$
are used only in our analysis;
furthermore, they depend on the random coins of \Cref{alg:offline_choose_fac}
(without any additional coins), 
hence our probabilistic analysis below refers to the random choices
made in \Cref{alg:offline_choose_fac}.

When dealing with sequences, we shall use the following notation. 
For two sequences $S'$ and $S''$, let $S' \circ S''$ denote their concatenation,;
when $S''=(x)$ is a singleton sequence we write $S'\circ x$,
which is a shorthand for appending the point $x$ to the sequence $S'$.
In addition, let $S' \sqsubseteq S''$ to denote that $S'$ is a prefix of $S''$.
We shall use $S'$, $S''$ and $T$ to denote fixed sequences,
reserving $S$ for the random sequence constructed by \Cref{alg:def_seq};
for instance, $S' \sqsubseteq S$ denotes an event,
namely, that \Cref{alg:def_seq} initially follows the given sequence $S'$,
after which it may or may not proceed further.

\begin{algorithm}[H]
\caption{Finding an assignment sequence $S = (x_0 = p, \ldots, x_t)$ for a given $p \in P$}
\label{alg:def_seq}
\begin{algorithmic}[1]
	\Require{point $p \in P$ and for every $q \in P$, $A^\beta_P(q, \hat{r}_q)$ as in \Cref{alg:offline_choose_fac}}
	\State let $i \gets 0$, $x_0 \gets p$, $S \gets (x_0)$ 
	\While{(P1) selects no point from $A_P^\beta(x_i, \hat{r}_{x_i})$ and (P2) does not select $x_i$} \label{line:while}
	
\State let $x_{i+1}$ be the point from $A_P^\beta(x_i, \hat{r}_{x_i})$ with the smallest label \label{line:next_point}
	\State let $S \gets S \circ x_{i+1}$ and $i \gets i + 1$
	\EndWhile
\end{algorithmic}
\end{algorithm}

The next lemma shows that \Cref{alg:def_seq} is well-defined
and always outputs a sequence of length at most $n$. 
\begin{lemma}
\label{lemma:terminate}
If $x_i$ passes the test in line~\ref{line:while},
then line~\ref{line:next_point} must find $x_{i + 1}\in P$
with $h(x_{i + 1}) < h(x_i)$.
This implies that \Cref{alg:def_seq} terminates after at most $n$ iterations.
\end{lemma}

\begin{proof}
Since $x_i$ passes the test, 
it does \emph{not} have the smallest label in $A^\beta_P(x_i, \hat{r}_{x_i})$.
Then in line~\ref{line:next_point},
$x_{i + 1}$ is set to the point with the smallest label in $A^\beta_P(x_i, \hat{r}_{x_i})$,
which certainly contains $x_i$.
Thus, $x_{i + 1}$ exists and has a smaller label than $x_i$,
which implies that the number of iterations is at most $|P|=n$. 
\end{proof}

\begin{definition}
   \label{def:weight_local_seq}
   The \emph{weight} of a sequence $S' = (x_0,\ldots,x_t)$ is 
$w(S') := \sum_{i=0}^t  \hat{r}_{x_i}$.
   A single point $p$ is viewed as a singleton sequence,
   and its weight is $w(p) = \hat{r}_p$.
We call $S'$ \emph{local} if
   for all $i \geq 1$, $x_i \in A^\beta_P(x_{i-1}, \hat{r}_{x_{i-1}})$.
\end{definition}

We can now bound the connection cost via $w(S)$,
where $S = (x_0 = p, \ldots, x_t)$ is the sequence generated by \Cref{alg:def_seq}. 
First, observe that $S$ is local, and thus 
\begin{equation}
\label{eq:local_seq}
  \forall i \geq 1, \qquad
  \dist(x_i, x_{i - 1}) \leq \beta \hat{r}_{x_{i  - 1}}.
\end{equation}
Second, the stopping condition of \Cref{alg:def_seq} implies that 
the last point $x_t$ must satisfy that
$A^\beta_P(x_t, \hat{r}_{x_t}) \cap F \neq \emptyset$ due to (P1)
or that $x_t\in F$ due to (P2);
in both cases, $A^\beta_P(x_t, \hat{r}_{x_t}) \cap F \neq \emptyset$.
Therefore, the connection cost of $p$ can be bounded by
\begin{align}
   \dist(p, F)
   \leq \dist(p, x_t) + \beta \hat{r}_{x_t}
   &\leq \sum_{i=1}^t \dist(x_{i-1}, x_i) + \beta \hat{r}_{x_t} 
   \leq \beta \sum_{i=0}^t \hat{r}_{x_i}
   = \beta \cdot w(S),
\label{eqn:exp_conn}
\end{align}
where the second inequality uses the triangle inequality,
and the third one uses~\eqref{eq:local_seq}.

\paragraph{Bounding the Connection Cost via the Critical Subsequence.}
We can bound the weight of a local sequence (and in particular of $S$)
using a certain subsequence, defined below as its \emph{critical subsequence}.
Informally, this subsequence skips steps where the $\hat{r}_p$ value is decreasing,
which by our assumptions means decreasing by factor $2$ or more;
it follows that every contiguous subsequence that is skipped
must have geometrically decreasing weights,
and thus its total weight is bounded by the weight of its immediately preceding point,
which is a part of the critical subsequence.
\begin{definition}[Critical Subsequence]
\label{def:crit_seq}
The \emph{critical subsequence} of a sequence $S' = (x_0 = p, \ldots, x_t)$ is 
   $\crit(S') := (x_i:\ i \geq 1 \text{ and } \hat{r}_{x_i} \geq \hat{r}_{x_{i-1}} )$.
\end{definition}

\begin{lemma}
\label{lemma:crit}
For every local sequence $S'=(x_i)_i$,
we have $w(S')^z \leq 2^{2z-1}  [ \hat{r}_{x_{0}}^z + w(\crit(S'))^z ]$.
\end{lemma}

\begin{proof}
Partition $S'$ into maximal contiguous subsequences
such that each element in $\crit(S')$ starts a new subsequence.
Thus, in each contiguous subsequence $T = (x_i, \ldots)$ the weights are decreasing
(note that $T$ may be singleton, which is considered decreasing). 
By our assumption that $\hat{r}_{p} < \hat{r}_q$ implies $\hat{r}_p \leq \hat{r}_q/2$,
we obtain that 
$w(T) \leq \hat{r}_{x_i} \cdot \sum_{j \geq 0} 2^{-j}\leq 2 \hat{r}_{x_i}$,
and aggregating over all such subsequences $T$, we have 
\[
  w(S') \leq 2 [\hat{r}_{x_0} + w(\crit(S))]. 
\]
The lemma follows by the generalized triangle inequality. 
\end{proof}

Using \Cref{lemma:crit}, we see that \eqref{eqn:exp_conn} implies that 
\begin{equation}
   \label{eqn:w_dist_to_crit}
  \dist^z(p, F)
  \leq \beta^z \cdot w(S)^z
  \leq \beta^z \cdot 2^{2z-1}  [ \hat{r}_{x_{0}}^z + w(\crit(S))^z ]. 
\end{equation}
It remains to bound $\E[w(\crit(S))^z]$,
for which we need the following notation. 
For a sequence $S'$, define
\begin{equation} \label{eqn:ellSprimeDefn}
  \ell(S') := \E[w(\crit(S)\setminus \crit(S'))^z \mid S' \sqsubseteq S] ;
\end{equation}
To understand this,
recall that $S$ denotes the random sequence generated by \Cref{alg:def_seq},
and $S' \sqsubseteq S$ is the event that the algorithm initially follows the given sequence $S'$,
hence $\ell(S')$ is informally the expected ``extra cost'' of the critical subsequence
generated by the algorithm (extra means after following $S'$),
where we use that 
$w(\crit(S)\setminus \crit(S')) = w(\crit(S)) - w(\crit(S'))$.
For $i=0,\ldots,n$ and $r > 0$, 
let $\mathcal{S}_{i, r}$ be the set of all sequences $S'$ (of points in $P$) that
(a) start at the point $p$;
(b) have $|\crit(S')| = i$; and
(c) end at any point $x'$ with $\hat{r}_{x'} \leq r$. 
Notice that by the definition of a critical subsequence, 
$\mathcal{S}_{o,r}$ (i.e., for $i=0$) can contain only sequences
with decreasing $\hat{r}$ value.
Next, define
\[
   \ell(i, r) := \max_{S' \in \mathcal{S}_{i, r}} \ell(S') ,
\]
which is the maximum expected extra cost over all prefixes in $\mathcal{S}_{i, r}$.
Observe that bounding $\ell(i, r)$, even only for $i=0$, would be useful because
the trivial sequence $S_0 := (x_0 = p)\in \mathcal{S}_{0, \hat{r}_p}$
has $\crit(S_0) = \emptyset$,
and therefore 
\begin{equation}
   \label{eqn:Eto_ell0}
   \E[w(\crit(S))^z]
   = \E[w(\crit(S) \setminus \crit(S_0)) ^z \mid S_0 \sqsubseteq S] 
   = \ell(S_0)
   \leq \ell(0, \hat{r}_p).
\end{equation}

\paragraph{Bounding $\ell(0, r)$.}
We shall prove that $\ell(i, r) \leq 2^z r^z$ for all $i$ and $r$ by induction.
More precisely, we shall bound $\ell(i,r)$ using a recursive formula,
which is essentially a weighted sum of $\ell(i,r/2)$ and $\ell(i+1,2r)$,
hence the induction will be on all pairs $(n-i,r)$, ordered lexicographically.
While the upper bound $2^z r^z$ is independent of $i$,
including $i$ in the inductive hypothesis is useful, and even crucial,
as induction only on $r$ would cause circular dependence
(the case of $r$ goes to $r/2$ and $2r$, and then back to $r$).
Informally, our main observation is that the coefficient of $\ell(i+1,2r)$
is a small constant $1/\tau$ (\Cref{lemma:crit_prob_cond}),
and thus when the recursive formula is expanded repeatedly,
the contribution of terms $\ell(i',\cdot)$ decays as $i'$ increases,
and the dominant contribution comes from the one term where $i'=i$.

To establish the recursive formula that bounds $\ell(i, r)$,
fix $i$, $r$, and a sequence $S' \in \mathcal{S}_{i, r}$.
Let $x'$ be the last point in $s'$, then $\hat{r}_{x'}\leq r$. 
Assuming that \Cref{alg:def_seq} initially follows the sequence $S'$,
formalized by conditioning on $S' \sqsubseteq S$,
and now the algorithm can proceed in three possible ways,
depending on its next iteration:
\begin{itemize}
\item ``stop'':
  The algorithm terminates,
  hence no more points are appended to the current sequence.
  This is formalized by $S = S'$, because that $S$ denotes the final sequence
  and by the conditioning $S'$ is the algorithm's ``current'' sequence. 
  It follows that $\crit(S)\setminus \crit(S'))$ is empty and has weight $0$,
  so this case contributes $0$ to the expected extra cost,
  regardless of its probability.
\item ``stay'':
  The algorithm's next iteration appends the current sequence
  with a point $x''$ that does not change its critical subsequence,
  which is formalized by $S' \circ x'' \sqsubseteq S$
  and $|\crit(S' \circ x'')| = |\crit(S')|$.
\item ``extend'':
  The algorithm's next iteration appends the current sequence
  with a point $x''$ that extends its critical subsequence, 
  which is formalized by $S' \circ x'' \sqsubseteq S$
  and $|\crit(S' \circ x'')| = |\crit(S')| +1$.
\end{itemize}
Define $\Pstay := \{ x'' \in P : |\crit(S' \circ x'')| = |\crit(S')| \}$
and $\Pextend := \{ x'' \in P : |\crit(S' \circ x'')| = |\crit(S')| + 1  \}$.
These two sets form a partition of $P$ 
and may include points that are too far from the last point of $S'$,
and thus have zero probability to be selected by the algorithm's next iteration. 
We can thus expand the definition in~\eqref{eqn:ellSprimeDefn} and write
\begin{equation} \label{eqn:ellSprime}
  \ell(S')
  = \E[w(\crit(S) \setminus \crit(S'))^z \mid S' \sqsubseteq S]
  = \ellstay(S') +  \ellextend(S'),
\end{equation}
where 
\begin{align*}
  \ellstay(S')
  &:= \sum_{x'' \in \Pstay}
    \Pr\big[S' \circ x'' \sqsubseteq S \mid S' \sqsubseteq S\big] \cdot
    \E\big[w(\crit(S) \setminus \crit(S' \circ x''))^z \mid S' \circ x'' \sqsubseteq S\big] ,
  \\
  \ellextend(S')
  &:= \sum_{x'' \in \Pextend}
    \Pr\big[S' \circ x'' \sqsubseteq S \mid S' \sqsubseteq S\big] \cdot
    \E\big[(w(\crit(S) \setminus \crit(S' \circ x'')) + \hat{r}_{x''})^z \mid S' \circ x'' \sqsubseteq S\big] .
\end{align*}

\paragraph{Bounding $\ellstay(S')$.}
Consider $x'' \in \Pstay$ with a nonzero probability to be selected.
Then $|\crit(S')| = |\crit(S' \circ x'')|$,
and recalling that $x'$ denotes the last point in $S'$, 
we get $\hat{r}_{x''} < \hat{r}_{x'}$
and in fact $\hat{r}_{x''} \leq \hat{r}_{x'} / 2 \leq r / 2$.
This implies that $S' \circ x'' \in \mathcal{S}_{i, r / 2}$ and 
\[
 \E[w(\crit(S) \setminus \crit(S' \circ x''))^z \mid S' \circ x'' \sqsubseteq S]  
 = \ell(S' \circ x'')
 \leq \ell(i, r / 2).
\]
By the trivial bound
$\sum_{x'' \in \Pstay}\Pr[S' \circ x'' \sqsubseteq S \mid S' \sqsubseteq S] \leq 1$,
we obtain 
\begin{equation}
   \label{eqn:ellstay}
   \ellstay(S') \leq \ell(i, r / 2).
\end{equation}

\paragraph{Bounding $\ellextend(S')$.}
Consider $x'' \in \Pextend$ with a nonzero probability to be selected,
hence $S' \circ x''$ is local. 
Then using \Cref{lemma:rp_ub}, 
\begin{align*}
   \hat{r}_{x''}
   \leq \alpha r_{x''}
   \leq \alpha 2^{(z - 1) / z} (r_{x'} + \dist(x', x'')) 
   \leq \alpha 2^{(z - 1) / z} (\hat{r}_{x'} + \beta \hat{r}_{x'})
   \leq \alpha 2^{(z - 1) / z} (\beta + 1) r .
\end{align*}
For the sake of brevity, let
\begin{equation}
   \label{eqn:eta}
   \eta := \alpha 2^{(z - 1) / z} (\beta + 1),
\end{equation}
so $\hat{r}_{x''} \leq \eta r$.
Therefore,
\begin{align}
  \E[ & (w(\crit(S) \setminus \crit(S' \circ x'')) + \hat{r}_{x''})^z \mid S' \circ x'' \sqsubseteq S] \nonumber \\
   &\leq 2^{z - 1} \cdot \E[w(\crit(S) \setminus \crit(S' \circ x''))^z + \hat{r}_{x''}^z \mid S' \circ x'' \sqsubseteq S] \nonumber \\
&\leq 2^{z - 1} \cdot \big[ \ell(i + 1, \eta r) + \eta^z r^z \big] , \label{eqn:extend_exp}
\end{align}
where the first inequality uses the generalized triangle inequality
and the second one uses the definition of $\ell(i + 1, \eta r)$.
We now need the following lemma,
whose proof appears in \Cref{sec:proof_lemma_crit_prob}.

\begin{restatable}{lemma}{lemmacritprobcond}
\label{lemma:crit_prob_cond}
Let $S'$ be a sequence that starts from $x_0 = p$.
Then
\begin{equation}
  \label{eqn:extend_prob}
  \sum_{x'' \in \Pextend} \Pr[S' \circ x'' \sqsubseteq S \mid S' \sqsubseteq S] \leq 1 / \tau.
\end{equation}
\end{restatable}

Putting together \eqref{eqn:extend_exp} and \eqref{eqn:extend_prob},  we obtain
\begin{equation}
  \label{eqn:ellextend}
  \ellextend(S') \leq 2^{z-1} \cdot \big[ \ell(i + 1, \eta r) + \eta^z r^z \big] / \tau.
\end{equation}

\paragraph{Bounding $\ell(i, r)$ Recursively.}
Plugging our bounds from~\eqref{eqn:ellstay} and~\eqref{eqn:ellextend}
into \eqref{eqn:ellSprime}, we obtain 
\begin{equation}
   \ell(i, r)
   = \max_{S' \in \mathcal{S}_{i,r}} \ell(S')
   \leq \ell(i, r / 2) + 2^{z - 1} \cdot \big[ \ell(i + 1, \eta r)  +  \eta^z r^z  \big] / \tau.
\end{equation}

We are ready to prove by induction that $\ell(i, r) \leq 2^z r^z$ for all $0\le i \leq n$ and all $r > 0$.
As mentioned earlier, the induction will be on all pairs $(n-i,r)$, ordered lexicographically.
In the base case $i = n$, \Cref{alg:def_seq} already terminates
and therefore $\ell(n, r) = 0$ for all $r > 0$.
For the inductive step,
assume that the induction hypothesis holds for all pairs that precede $(i,r)$ in the above order,
and let us verify it for $(i,r)$, as follows. 
\begin{align*}
   \ell(i, r)
   &\leq \ell(i, r / 2) + 2^{z - 1} \cdot \big[ \ell(i + 1, \eta r)  +  \eta^z r^z  \big] / \tau \\
   &\leq r^z + 2^{z - 1} \cdot [2^z \eta^z r^z + \eta^z r^z] / \tau \\
   &= r^z (1 + 2^{z - 1}(2^z \eta^z + \eta^z) / \tau)  \\
   &\leq 2^z r^z,
\end{align*}
where the last inequality holds by setting 
$\tau \geq 2^{2z}\eta^z = (\Theta(1)\alpha\beta)^z$,
recalling the value of $\eta$ from~\eqref{eqn:eta}.

\paragraph{Concluding \Cref{lemma:offline_fac_finding}.}
Putting together~\eqref{eqn:Eto_ell0} and the above bound
$\ell(i, r) \leq 2^z r^z$,
we have
\[
   \E[w(\crit(S))^z]
   \leq \ell(0, \hat{r}_p)
   \leq 2^z \hat{r}_p^z.
\]
Plugging this bound into~\eqref{eqn:w_dist_to_crit},
the expected connection cost from $p$ is 
\begin{align*}
  \E[\dist^z(p, F)]
  \leq \beta^z 2^{2z - 1} [ \hat{r}^z_{p} + \ell(0, \hat{r}_{p}) ]
  \leq \beta^z 2^{2z - 1} [ \hat{r}^z_{p} + 2^z \hat{r}^z_{p} ]
  \leq \beta^z 2^{O(z)} \cdot \hat{r}^z_{p}.
\end{align*}
Thus, the expectation of the total connection cost is at most
\[
   \sum_{p\in P} (\beta^z 2^{O(z)}\cdot \alpha^z  r_p^z)
   \leq 2^{O(z)}\alpha^z\beta^z \cdot \OPTfl_z,
\]
which proves \Cref{lemma:offline_fac_finding}.
\end{proof}

\subsubsection{Proof of \Cref{lemma:crit_prob_cond}}
\label{sec:proof_lemma_crit_prob}

\lemmacritprobcond*

To prove this lemma, which examines a single iteration of \Cref{alg:def_seq},
let us rewrite the lefthand-side of~\eqref{eqn:extend_prob}, as follows. 
Fix $S' = (x_0 = p, \ldots, x_t)$ as before,
and let $\Econd$ be the event that $S' \sqsubseteq S$,
i.e., \Cref{alg:def_seq} gets to line~\ref{line:while}
with the current sequence equal to $S'$.
We shall be conditioning on $\Econd$ throughout this proof,
but to avoid confusion we shall often write this conditioning explicitly.
Let $\Eextend$ be the event that $\Econd$ holds and
in iteration $i=t$ of \Cref{alg:def_seq}, 
$x_t$ passes the test in line~\ref{line:while} 
and the point picked in line~\ref{line:next_point} is from $\Pextend$.
Observe that the lefthand-side of~\eqref{eqn:extend_prob}
can be written as
$\Pr[ \Eextend \mid \Econd ] $, 
and it remains to bound this quantity.

Observe that \Cref{alg:offline_choose_fac} makes two random choices
for each point $q\in P$, 
one is whether to open $q$ by (P1), 
and the other is the random label $h(q)$ used in (P2). 
The event $\Econd$ may depend on the random choices for points in the set
$\Acond := \bigcup_{i = 0}^{t - 1} A_P^\beta(x_i, \hat{r}_{x_i})$,
but not on those for points outside $\Acond$.
The choices for points outside $\Acond$ and inside $\Acond$ are independent,
hence by the principle of deferred decisions,
conditioning on $\Econd$ does not change the distribution of choices for points in
$\Anew := A_P^\beta(x_{t}, \hat{r}_{x_t}) \setminus \Acond$,
which in turn affect the event $\Eextend$ (e.g., whether $x_t$ passes the test).
Since points with $\hat{r}$ value at least $\hat{r}_{x_t}$ are of particular importance,
we also define 
$\Anewgeq := \{ q \in \Anew:\ \hat{r}_q \geq \hat{r}_{x_t} \}$.

We can write $\Eextend|\Econd$ as the intersection of two events, as follows.
Define $\EPone$ to be the event that no point in $A_P^\beta(x_t, \hat{r}_{x_t})$ is selected by (P1).
Define $\EPtwo$ to be the event that $x_t $ is not selected by (P2)
and the point with the smallest label in $A_P^\beta(x_t, \hat{r}_{x_t})$,
denoted $x_{t + 1}$ is in $\Pextend$, i.e., $\hat{r}_{x_{t+1}} \geq \hat{r}_{x_t}$.
Strictly speaking, in both definitions (just as in $\Eextend$), 
we also require that $\Econd$ holds,
and therefore $\Eextend=\EPone\land\EPtwo$.
We claim that $\EPone$ and $\EPtwo$ are conditionally independent (given $\Econd$), i.e., 
\begin{align}
  \Pr[\Eextend \mid \Econd]
  = \Pr[\EPone \land \EPtwo \mid \Econd] 
  = \Pr[\EPone \mid \Econd] \cdot \Pr[\EPtwo \mid \Econd].
  \label{eqn:sub_events}
\end{align}
One can easily verify this claim, 
for instance by reading our analysis below of these two events;
in a nutshell,
$\EPone \mid \Econd]$ depends only on random choices made by (P1),
specifically regarding points outside $\Acond$,
while $\EPtwo \mid \Econd$ depends only on random labels used by (P2),
specifically the relative order between labels outside and inside $\Acond$.

\paragraph{Bounding $\Pr[\EPone | \Econd]$.}
When $\EPone|\Econd$ occurs, 
no point from $\Anew$ is selected by (P1),
which in turn is independent of $\Econd$,
because $\Anew$ is disjoint of $\Acond$.
Therefore,
\begin{align}
  \Pr[\EPone \mid \Econd]
  &\leq \Pr[\text{no point in $\Anew$ is selected by (P1)} \mid \Econd] \nonumber \\
  &= \Pr[\text{no point in $\Anew$ is selected by (P1)}] \nonumber \\
  &= \prod_{q \in \Anew} \big(1 - \tau\cdot \hat{r}_q^z / \wopen\big) \nonumber \\
&\leq \exp\big(- |\Anewgeq|\cdot \tau\cdot \hat{r}_{x_t}^z / \wopen\big) .
    \label{eqn:EP1}
\end{align}

\paragraph{Bounding $\Pr[\EPtwo | \Econd]$.}
When $\Econd$ occurs, 
each point $x_i$, for $i=1,\ldots,t$,
has the smallest label in $A_P^\beta(x_{i - 1}, \hat{r}_{x_{i-1}})$,
and it follows by induction that $x_t$ has the smallest label in $\Acond$.
When $\EPtwo$ occurs,
some $x_{t+1}\neq x_t$ has the smallest label in $A_P^\beta(x_t, \hat{r}_{x_t})$,
and thus $h(x_{t+1}) < h(x_t) = \min(h(\Acond))$,
implying that $x_{t+1} \notin \Acond$ and in fact $x_{t+1} \in \Anewgeq$.
Therefore,
\begin{align} \label{eqn:qAnew}
  \Pr[\EPtwo \mid \Econd]
  \leq \sum_{q \in \Anewgeq} \Pr\left[ h(q) = \min(h(A_P^\beta(x_t, \hat{r}_{x_t})))  \bigm| \Econd \right]. 
\end{align}
The conditioned events in the summation depend only on the randomness of $h$,
and certainly not on random choices made by (P1).
In addition, the conditioning on $\Econd$
effectively provides only a partial-order information about $h$,
and says that for each $i=1,\ldots,t$ we have
$h(x_i) = \min(h(A_P^\beta(x_{i - 1}, \hat{r}_{x_{i-1}})))$.
The next lemma captures this scenario using a simplified probabilistic model. 

\begin{claim}
\label{lemma:partial_order}
Fix a finite domain $\mathcal{U}$,
and let $g : \mathcal{U} \to [0, 1]$ be random, 
such that each $g(a)$ is chosen independently and uniformly from $[0,1]$.
Then for every subset $S \subseteq \mathcal{U}$, element $a \in S$,
and partial order $\mathcal{P}$ on $S \setminus \{a\}$, we have
\[
  \Pr_g[g(a) = \min(g(S)) \mid g \POconsistent \mathcal{P}] = 1 / |S|,
\]
where $g \POconsistent \mathcal{P}$ denotes consistency
in the sense that $g(x) < g(y)$ whenever $x \prec_{\mathcal{P}} y$.
\end{claim}

\begin{proof}
Clearly, all permutations of $g(S)$ have the same probability. 
Now observe that if we expand the conditioned event to reveal
the complete ordering of $S\setminus\{a\}$,
then the probability that $g(a) = \min(g(S))$ is $1/|S|$,
by counting the relevant permutations of $g(S)$.
The claim follows by the law of total probability.
\end{proof}
    
We can bound $\Pr[ h(q) = \min(h(A_P^\beta(x_t, \hat{r}_{x_t}))) \mid \Econd ]$
for a given $q \in \Anewgeq$
by applying \Cref{lemma:partial_order}
with $g := h$, $a := q$,
$S := \Anew \cup \Acond = \bigcup_{i = 1}^t A_P^\beta(x_t, \hat{r}_{x_t})$
which indeed contains $a$, 
and a partial order $\mathcal{P}$ defined from the event $\Econd$,
namely, that for each $i=1,\ldots,t$ we have
$h(x_i) = \min(h(A_P^\beta(x_{i - 1}, \hat{r}_{x_{i-1}})))$,
noticing that this is indeed a partial order on $S\setminus \{a\}$.
Hence, 
\[
  \forall q \in \Anewgeq,
  \qquad 
  \Pr[ h(q) = \min(h(A_P^\beta(x_t, \hat{r}_{x_t}))) \mid \Econd ]
    = \frac{1}{|\Anew \cup \Acond|}
    \leq \frac{1}{|A_P^\beta(x_t, \hat{r}_{x_t})|}.
\]
Plugging this into \eqref{eqn:qAnew}, we obtain
\begin{equation}
  \label{eqn:EP2}
  \Pr[\EPtwo \mid \Econd]
  \leq \frac{|\Anewgeq| }{ |A_P^\beta(x_t, \hat{r}_{x_t})| }.
\end{equation}

\paragraph{Putting It Together.}
We now use \eqref{eqn:sub_events}, \eqref{eqn:EP1}, and~\eqref{eqn:EP2}
to obtain 
\begin{align*}
  \Pr[\mathcal{\Eextend} \mid \Econd]
  &= \Pr[\EPone \mid \Econd] \cdot \Pr[\EPtwo \mid \Econd] \\
  &\leq \exp\big(- |\Anewgeq|\cdot \tau\cdot \hat{r}_{x_t}^z / \wopen \big)
    \cdot \frac{ |\Anewgeq| }{|A_P^\beta(x_t, \hat{r}_{x_t})|} \\
  &\leq \exp\Big(- \tau\cdot \frac{|\Anewgeq|}{|B_P(x_t, r_{x_t})|} \Big)
    \cdot \frac{ |\Anewgeq| }{|B_P(x_t, r_{x_t})|} 
    \leq \frac{1}{\tau},
\end{align*}
where the second inequality uses \Cref{fact:rp}
and that $\hat{r}_q \geq r_q$ for all $q \in P$,
and the last inequality uses that $x \exp(-x) \le 1$ for all $x\in\RR$. 
This concludes the proof of \Cref{lemma:crit_prob_cond}.
\qed

\subsection{MPC Implementation: Proof of \Cref{thm:ufl}}
\label{sec:mpc_ufl}

The MPC algorithm for \Cref{thm:ufl} consists of two MPC implementations,
one for \Cref{lemma:rp_offline} that we present in \Cref{alg:mpc_rp},
and one for \Cref{alg:offline_choose_fac} that we present in \Cref{alg:mpc_fac}.
Both MPC implementations tightly follow their sequential counterpart,
with the key steps implemented using the MPC primitive provided by \Cref{thm:mpc}.

\begin{algorithm}[ht]
   \caption{MPC implementation of \Cref{lemma:rp_offline}, to compute $\hat{r}_p$ that approximates $r_p$ for each $p \in P$}
   \label{alg:mpc_rp}
   \begin{algorithmic}[1]
     \State let $\beta \gets 3\Gamma$ as in \Cref{thm:mpc}
     \State let 
      $Z \gets \{ \beta^i :\ i \in \mathbb{Z} \text{ and } \wopen / (3\beta|P|)\leq \beta^{iz} \leq  \beta \wopen\}$
      \label{line:defZ}
      
        \Comment{all possible values of $\hat{r}_p$, per \eqref{eqn:rp_range} and \Cref{lemma:rp_offline} }

      \State for each $p \in P$ and $r \in Z$, compute $|A^\beta_P(p, r)|$ where $B_P(p, r) \subseteq A^\beta_P(p, r)  \subseteq B_P(p, \beta r)$

      \Comment{done \emph{in parallel} for all $r \in Z$, using \Cref{thm:mpc} with $f:S\mapsto |S|$}
      \State for each $p \in P$, find the smallest $\hat{r}_p \in Z$ such that $|A^\beta(p, \hat{r}_p)| \geq \wopen / (2\beta^z \hat{r}_p^z)$ as in \Cref{lemma:rp_offline} 

      \Comment{done by sorting $|A^\beta(p, r)|$ by $p$}
      \State for each $p \in P$, rescale $\hat{r}_p \gets \hat{r}_p \cdot 3\beta$
   \end{algorithmic}
\end{algorithm}

\begin{algorithm}[ht]
   \caption{MPC implementation of \Cref{alg:offline_choose_fac}, to compute a facility set $F\subseteq P$ }
   \label{alg:mpc_fac}
   \begin{algorithmic}[1]
      \Require{$\hat{r}_p \in [r_p, \alpha r_p]$ for each $p \in P$, resulting from \Cref{alg:mpc_rp}}

      \State let $\beta \gets 3\Gamma$ as in \Cref{thm:mpc}, 
        and round each $\hat{r}_p$ to the next integral power of $2$
      \State let $F \gets \emptyset$, and 
      let $\tau  \gets  (\alpha \beta)^{\Theta(z)} $ 
        
      \For{each machine $\mathcal{M}$ and each point $p$ stored in $\mathcal{M}$, in parallel}
         \State assign a uniformly random label $h(p)\in[0,1]$ and store it with $p$
         \State add $p$ to $F$ with probability $\tau \cdot \hat{r}_p^z / \wopen$ (independently of $h$)
      \EndFor
      \State for each $p \in P$ and $r \in \{ \hat{r}_q : q \in P \}$ 
      compute $\hat{h}_{\min}^{r}(p) \gets \min(h(A_P^{\beta}(p, r)))$
\label{line:mpc_hmin}

      \Comment{done in parallel for all possible $r$, using \Cref{thm:mpc} with $f:S\mapsto \min(h(S))$ }

      \State for each $p \in P$,
        let $\hat{h}_{\min}(p) \gets \hat{h}_{\min}^{\hat{r}_p}(p)$
        and if $\hat{h}_{\min}(p) = h(p)$ then add $p$ to $F$
   \end{algorithmic}
\end{algorithm}

The overall algorithm (for \pzFL) starts with executing \Cref{alg:mpc_rp}
to compute the $\hat{r}_p$ values.
By \Cref{lemma:rp_offline} and \Cref{thm:mpc}, the computed values
satisfy that $r_p \leq \hat{r}_p \leq 3 \beta^2 = O(\Gamma^2)$,
where $\Gamma$ is the consistent-hashing parameter used in \Cref{thm:mpc}.

Next, the overall algorithm executes \Cref{alg:mpc_fac}
to compute a set denoted $F\subset P$ of open facilities,
where we set $\alpha \eqdef 3\beta^2 = \poly(\Gamma)$
as in \Cref{lemma:rp_offline}. 
Notice that this algorithm is an MPC implementation of \Cref{alg:offline_choose_fac}.
Hence by \Cref{{lemma:offline_fac_opening},lemma:offline_fac_finding},
as well as \Cref{thm:mpc} and the above guarantee about \Cref{alg:mpc_rp},
we conclude that
\[
   \E[\fl_z(P, F) ] \leq \Gamma^{O(z)} \cdot \OPTfl_z.
\]

To achieve the required $1 / \poly(n)$ failure probability bound,
one can use a standard amplification trick,
by running the algorithm for $O(\log n)$ times in parallel
and taking one with minimum value.
However, this method requires evaluating the objective for the found facilities,
which is nontrivial, since the number of facilities might be large
(e.g., not fit into a single machine),
but one can apply \Cref{thm:mpc} again to find an approximate nearest-neighbor
(facility) in parallel for each data point $p\in P$
(a more detailed description appears in \Cref{sec:nearest-neighbor-search}),
and then aggregate their distances (their power-$z$ sum).
This additional step increases the round complexity and total space 
by $O(\log n)$ factor, which is still within the bounds of \Cref{thm:ufl}.
This finishes the analysis for the approximation ratio and failure probability.

Finally, for the round complexity and total space,
observe that in both \Cref{alg:mpc_rp,alg:mpc_fac},
these are dominated by the parallel invocations of \Cref{thm:mpc},
and the number of parallel invocations in \Cref{alg:mpc_rp} is
$|Z| \leq O(\log n)$, and in \Cref{alg:mpc_fac} it is 
$| \{\hat{r}_p : p \in P\} | \leq |Z| \leq O(\log n)$. All other steps can be easily implemented in $O(\log_{s}n)$ rounds
with at most $\polylog (n)$ factor of additional space.
This completes the proof of \Cref{thm:ufl}.

 \section{Clustering}
\label{sec:clustering}

In this section, we develop an MPC algorithm for $\kzC$ that outputs a bicriteria approximate solution,
using only relatively few extra centers compared to the optimal solution.
Specifically, given $\mu > 0$,
the algorithm computes a \kzC solution that uses at most $(1+\mu)k$ centers 
and has connection cost at most $(\eps^{-O(z)}\mu^{-2})$-times larger than the optimal cost $\OPTkz$,
using total space $O(n^{1+\eps})$ for any $\eps > 0$;
the trade-off between the total space and the approximation is given by the parameters of consistent hashing,
similarly as for \pzFL in \Cref{thm:ufl}.
We stress that this result is most interesting when $k$ substantially exceeds the local memory $s$ in each machine.
We give a more precise statement, where the guarantees
are expressed in terms of the
consistent-hashing parameters $\Gamma$ and $\Lambda$ used by the algorithm.

\begin{theorem}
\label{thm:clustering}
There is a randomized fully scalable MPC algorithm for \kzC (for some $z\geq 1$)
with the following guarantees.
Suppose that
\begin{itemize} \item the input is an integer $k\geq 1$, parameter $\mu\in (0,1)$, aspect ratio $\Delta $ of $P$,
  and a multiset $P\subset \RR^d$ of $n$ points
  distributed across machines with local memory $s\geq \poly(d\log n)$; and
\item the algorithm has access
to the same $(\Gamma, \Lambda)$-hash $\varphi:\RR^d\to\RR^d$ in every machine. \end{itemize}
The algorithm uses $O(\log_{s}n)$ rounds
and $O(\Lambda \poly(d) \log(\Delta)) \cdot \Ot(n)$ total space,
and succeeds with probability at least $1-1/\poly(n)$
to compute a set $C\subseteq P$ of at most $(1+\mu)k$ centers,  
whose connection cost is 
$\cl_z(P, C) \le \Gamma^{O(z)}\mu^{-2}\cdot \OPTkz(P)$,
where $\OPTkz(P)$ is the optimal \kzC cost.
\end{theorem}

We develop the algorithm in the following three steps. 
Let $\gamma := \Gamma^{O(z)}$ be the approximation factor of
algorithm $\algoFLz$ for \pzFL (from \Cref{thm:ufl}). 
\begin{enumerate}
\item We use algorithm $\algoFLz$ to reduce the
  original instance with $n$ points to a \emph{weak coreset} $P_w$,
  that is, an instance with at most $2\gamma\cdot k$ distinct weighted points of total weight $n$, 
  such that the optimal cost of the new instance $P_w$ is
  $\OPTkz(P_w) \le O(\gamma)\cdot \OPTkz(P)$. 
\item We then design a \emph{sequential} algorithm for weak coresets.
  This algorithm computes a set $C$ of at most $(1+\mu) k$ centers
  that approximates the optimal connection cost (of the weak coreset)
  within factor $2^z \gamma / \mu$.
\item We provide a parallel version of the sequential algorithm,
  and show how to implement it in the MPC model.
  This parallel algorithm relies on similar ideas to those used
  in the facility location algorithm $\algoFLz$.
\end{enumerate}

\paragraph{Preprocessing.}
We assume without loss of generality that $P$ has more than $k$ distinct points.
Otherwise, clearly $\OPTkz(P) = 0$, 
and that in $O(\log_{s}n)$ rounds one can determine
whether or not this is the case by sorting and de-duplicating the input
(using the algorithm of Lemma~\ref{lemma:sorting}). 
Moreover, since $\OPTkz(P)$ is between the minimum inter-point distance in $P$ and $n \cdot \diam(P)$,
we can first compute a $2$-approximation of $\diam(P)$ in $O(\log_s n)$ rounds
(similarly as in \Cref{sec:nns}),
and rescale the point set so that $1 \le \OPTkz(P) \le n\cdot O(\Delta)$.

\subsection{Reduction to Weak Coresets}

We provide a reduction to an instance with bounded number of weighted points $P_w$.
For completeness, we define the weighted clustering cost $\cl_z(P_w, C)$ as the weighted analogue of $\cl_z(P, C)$ in the natural way:
$\cl_z(P_w, C) := \sum_{p \in P_w} \dist^z(p, C)\cdot w(p)$.

\begin{lemma}\label{lemma:clustering-reduction_to_weighted}
Suppose there is a $\gamma$-approximation MPC algorithm $\algoFLz$ for \pzFL
that succeeds with probability at least $1-1/\poly(n)$.
Then there is a randomized fully scalable MPC algorithm that,
given the same input as in Theorem~\ref{thm:clustering},
computes a set $P_w\subseteq P$ of at most $O(\gamma)\cdot k$ weighted points
such that
\begin{itemize} \item the total weight of points in $P_w$ is $n$, and
\item with probability at least $1-1/\poly(n)$,
	it holds that $\OPTkz(P_w) \le O(\gamma)\cdot \OPTkz(P)$.
\end{itemize}
Compared to $\algoFLz$, the algorithm uses $O(\log_{s}n)$ additional rounds
and its total space is factor $O(\log (n \Delta))$ larger than that of $\algoFLz$.
\end{lemma}

\begin{proof}
	Assume first we have foreknowledge of a guess $\OPTkzGuess$
	such that $\OPTkz(P)\le \OPTkzGuess\le 2\OPTkz(P)$. 
	Consider $P$ as an instance of Power-$z$ Facility Location with opening cost $\hat{\wopen} := \OPTkzGuess / k$,
	and note that $\OPT^{\hat{\wopen}}_z \le k\cdot \hat{\wopen} + \OPTkz(P) \leq 2\cdot  \OPTkzGuess$ as the optimal solution for \kzC is a particular
	solution to the Facility Location instance (using at most $k$ facilities).
	Now run $\algoFLz$ with this opening cost $\hat{\wopen}$, 
        and obtain a solution with cost at most
        $\gamma\cdot \OPT^{\hat{\wopen}}_z \le 2\gamma \cdot  \OPTkzGuess$\footnote{
        	We assume that the solution specifies an assignment of points to facilities.
        	Otherwise, given just a set of facilities (distributed across machines), we run
			an MPC algorithm for approximate nearest neighbor search,
			for instance, based on LSH or \Cref{thm:mpc}; see \Cref{sec:nearest-neighbor-search}.
			We assume that $\gamma$ accounts for the potential increase of the connection cost
			due to possibly not finding the closest facility to each point.
	    }.
	The number of facilities used by this approximate solution is clearly at most
	\[
          \frac{\gamma\cdot \OPT^{\hat{\wopen}}_z}{\hat{\wopen}}
          \le \frac{2\gamma\cdot \OPTkzGuess}{\OPTkzGuess / k}
          = 2\gamma\cdot k\,.
        \]
        Given this approximate solution,
        ``move'' every input point $p\in P$ to its assigned facility (which is approximately the nearest one)
	and give each facility weight equal to the number of points moved to it. 
	This can be implemented in $O(\log_{s}n)$ rounds by sorting (using the algorithm of Lemma~\ref{lemma:sorting}).
	Hence, the facilities form a set of at most $2\gamma\cdot k$ weighted points $P_w$ of total weight exactly $n$.
        Since the points are moved using an approximate solution 
        whose connection cost is at most $2\gamma \cdot \OPTkzGuess$, 
	we have by the triangle inequality that
$\OPTkz(P_w) \le 2\gamma \cdot \OPTkzGuess + \OPTkz(P) \le O(\gamma)\cdot \OPTkz(P)$.
			
	Finally, we deal with the assumption of knowing an approximation of $\OPTkz(P)$ in advance.
	Recall that we assume $P$ has more than $k$ distinct points, 
	and that $1 \le \OPTkz(P) \le O(n \Delta)$ where $\Delta$ is the aspect ratio of $P$.
	We try in parallel all powers of 2 between $1$ and $O(n \Delta)$ as ``guesses'' for $\OPTkz(P)$.
	That is, we execute $\algoFLz$ with opening cost $2^i /k$ for
	$i = 0, \ldots, \lfloor \log_2 O(n\Delta)\rfloor$, and
	we condition on that every call of $\algoFLz$ is successful (i.e., returns a solution satisfying the approximation guarantee),
	which happens with probability at least $1-1/\poly(n)$.
	We then choose the cheapest approximate solution that uses at most $2\gamma\cdot k$ facilities,
        which succeeds because at least one guess is in the desired range. 
	Such guessing of $\OPTkz$ requires a factor of $O(\log (n \Delta))$ extra total space.
\end{proof}

Note that if $O(\gamma\cdot k)$ is less than the machine space $s$, one can in principle
solve \kzC on the weak coreset using just one machine. Thus the main challenge is to do this
for large $k \gg s$.

\subsection{Sequential Algorithm for Weak Coresets}

Suppose we are given a set of points $P_w \subseteq \RR^d$ with weights $w : P_w \to \mathbb{R}_+$, and 
$|P_w| \leq \gamma \cdot k$ for some $\gamma \geq 1$.
Furthermore, assume we have a guess $\OPTkzGuess$ that $2$-approximates $\OPTkz(P_w)$,
i.e., $\OPTkz(P_w)\le \OPTkzGuess\le 2\cdot \OPTkz(P_w)$
(we deal with this assumption in Section~\ref{sec:clustering-parallel} for the parallel version,
similarly as in Lemma~\ref{lemma:clustering-reduction_to_weighted}).
Finally, the algorithm gets $\mu > 0$ and will ensure that at most $\mu\cdot k$ extra centers are used.

The sequential algorithm to find a set of centers $C$ is simple: 
We go over the points from the heaviest to the lightest and open a center at the current point $p$
if it is far from all previously opened centers. The right distance threshold
turns out to be $2 \big[ \OPTkzGuess / (\mu\cdot k\cdot w(p)) \big]^{1/z}$;
note that it is larger for points with smaller weight.
See Algorithm~\ref{alg:sequential_clustering} for details.

\begin{algorithm}[H]
	\caption{Sequential clustering algorithm for weak coresets}
	\label{alg:sequential_clustering}
	\begin{algorithmic}[1]
		\Require{point set $P_w$ weighted by $w : P_w \to \mathbb{N}$ and a guess $\OPTkzGuess$ for $\OPTkz(P_w)$}
		\State initialize $C \gets \emptyset$
		\For{each $p \in P_w$ in the non-increasing order by weight (breaking ties arbitrarily)}
			\If{$\dist(p, C)^z\cdot w(p) > 2^z \cdot \OPTkzGuess / (\mu\cdot k)$}
				\State add $p$ to $C$
			\EndIf
		\EndFor
	\end{algorithmic}
\end{algorithm}

Let $C$ be the set of centers computed by the sequential algorithm.
The next two lemmas bound the weighted clustering cost for $C$,
which is straightforward,
and the number of centers in $C$, which is the main part. 

\begin{lemma}\label{lemma:clustering-sequential-connection_cost}
  If the guess is correct then 
  $\cl_z(P_w, C) \le 2^{z+1}\gamma\cdot \OPTkz(P_w) / \mu$.
\end{lemma}

\begin{proof}
	Since every point $p\in P_w \setminus C$ satisfies $\dist(p, C)^z\cdot w(p) \le 2^z\cdot \OPTkzGuess / (\mu\cdot k)$
	and there are at most $\gamma \cdot k$ distinct points in $P_w$,
	we have that
	\[
          \cl_z(P_w, C) = \sum_{p \in P_w} \dist(p, C)^z\cdot w(p) \le \gamma \cdot k\cdot 2^z\cdot \OPTkzGuess / (\mu\cdot k) \le 2^{z+1}\cdot \gamma\cdot \OPTkz(P_w) / \mu,
        \]
	where the last inequality uses that $\OPTkzGuess$ is a $2$-approximation of $\OPTkz(P_w)$.
\end{proof}

\begin{lemma}\label{lemma:clustering-sequential-centers}
	$|C| < (1 + \mu) k$.
\end{lemma}

\begin{proof}
	Suppose otherwise, i.e., $|C| \ge (1 + \mu)\cdot k$.
	Consider an optimal \kzC solution of cost $\OPTkz(P_w)$ and a cluster $P_j\subseteq P_w$ of this solution centered at $q$.
	Let $\ell_j := |C\cap P_j|$ be the number of points from $C$ in cluster $P_j$.
	The goal is to show that the weighted connection cost of any cluster $P_j$ is more than $(\ell_j-1)\cdot \OPTkz(P_w) / (\mu\cdot k)$ in the optimal solution.
	This is sufficient to get a contradiction, 
	since
	\begin{align*}
		\sum_j (\ell_j-1)\cdot \frac{\OPTkz(P_w)}{\mu\cdot k}
		&\ge |C|\cdot \frac{\OPTkz(P_w)}{\mu\cdot k} - k\cdot \frac{\OPTkz(P_w)}{\mu\cdot k}
		\\
		&\ge (1 + \mu)\cdot k\cdot \frac{\OPTkz(P_w)}{\mu\cdot k} - k\cdot \frac{\OPTkz(P_w)}{\mu\cdot k}
		= \OPTkz(P_w)
		\,,
	\end{align*}
	where the first inequality uses that $\sum_j \ell_j = |C|$ and that there are at most $k$ clusters,
	and the second inequality follows from $|C| \ge (1 + \mu)\cdot k$.
	This implies that the cost of the optimal solution is more than $\OPTkz(P_w)$, a contradiction.

	To analyze the cost of cluster $P_j$,
	let $p_1, p_2, \dots, p_\ell$ be the points from $C$ in cluster $P_j$ in the non-increasing order by their weights, i.e.,
	$w(p_1)\ge w(p_2)\ge \cdots\ge w(p_\ell)$ (breaking ties in the same way as in the sequential algorithm).
	We claim that at most one point $p_i$ satisfies $\dist(p_i, q)^z\cdot w(p_i) \le \OPTkz(P_w) / (\mu\cdot k)$, i.e.,
	its weighted connection cost is too low.
	Suppose otherwise, i.e., that there are points $p_a, p_b \in P_j$ with $a < b$ that have such low weighted connection cost to $q$.
	Then
	\begin{align*}
		\dist(p_a, p_b)^z\cdot w(p_b)
		&\le 2^{z-1}\cdot \dist(p_a, q)^z\cdot w(p_b) + 2^{z-1}\cdot \dist(p_b, q)^z\cdot w(p_b)
		\\
		&\le 2^{z-1}\cdot \dist(p_a, q)^z\cdot w(p_a) + 2^{z-1}\cdot \dist(p_b, q)^z\cdot w(p_b)
		\\
		&\le 2^z\cdot \frac{\OPTkz(P_w)}{\mu\cdot k}
		\le 2^z\cdot \frac{\OPTkzGuess}{\mu\cdot k}\,,
	\end{align*}
	where the first inequality is by the generalized triangle inequality (Lemma~\ref{lemma:tri_ine_z}),
	the second step is by $w(p_a) \ge w(p_b)$ as $a < b$,
	the third inequality uses the assumptions on $\dist(p_a, q)^z\cdot w(p_a)$ and $\dist(p_b, q)^z\cdot w(p_b)$,
	and the last inequality follows from that $\OPTkzGuess$ is not an underestimation by assumption.
	This contradicts that the sequential algorithm opens a center at $p_b$, since the center at $p_a$ was added to $C$ earlier.
	
	Thus, the total connection cost of cluster $P_j$ is more than $(\ell_j-1)\cdot \OPTkz(P_w) / (\mu\cdot k)$, which is a contradiction
	with that the optimal clustering cost is $\OPTkz(P_w)$, as explained above.
\end{proof}

\subsection{Parallel Algorithm for Weak Coresets}
\label{sec:clustering-parallel}

In Algorithm~\ref{alg:offline_bi_criteria_clustering},
we give a parallel version of the sequential algorithm,
without MPC implementation details.
The input is again a weak coreset $P_w$ and the algorithm requires a guess $\OPTkzGuess$ of the optimal cost for $P_w$,
which we use to set parameter $\rho := 2^z \cdot \OPTkzGuess / (\mu\cdot k)$.
We subsequently analyze the algorithm, assuming that $\OPTkz(P_w)\le \OPTkzGuess\le 2\cdot \OPTkz(P_w)$.
We then remove the assumption of having such a guess $\OPTkzGuess$.
Finally, we show how to implement Algorithm~\ref{alg:offline_bi_criteria_clustering} in the MPC model.
These steps rely on similar ideas as we used for Facility Location, adjusted to the setting
of \kzC for a bounded number of weighted points.

We use rules (C1) and (C2) similar to rules (P1) and (P2) in Algorithm~\ref{alg:offline_choose_fac} for Facility Location.
We note the crucial differences: (C1) opens a center at any point with \emph{constant} probability $\mu / \gamma$,
unlike in (P1) where the probability depends on the $r_p$ value.
The constant probability is possible here as our instance is a coreset with bounded number of distinct points.
For (C2), instead of looking for the smallest label in a neighborhood of radius roughly $r_p$,
we find the maximum weight in a ball of radius $\sqrt[z]{\rho / w(p)}$; thus, 
the weight can be thought of as an analogue of the inverse $r_p$ value.

We note that we perturb the weights at random to $w'(p)$ in line \ref{ln:alg-clustering-parallel-perturb} of Algorithm~\ref{alg:offline_bi_criteria_clustering}
to resolve ties between same-weight points in a consistent and easy-to-parallelize way (without affecting the original partial order).
We assume that for any $p\neq q$, $w'(p)\neq w'(q)$.

To efficiently implement rule (C2), we again require a method that finds 
for every point $q \in P_w$,
the minimum perturbed weight in a set $A^\beta_{P_w}(p, \sqrt[z]{\rho / w(p)})$
that approximates $ B_{P_w}(p, \sqrt[z]{\rho / w(p)})$,
specifically
$ B_{P_w}(p, \sqrt[z]{\rho / w(p)}) \subseteq A^\beta_{P_w}(p, \sqrt[z]{\rho / w(p)}) \subseteq B_{P_w}(p, \beta\sqrt[z]{\rho / w(p)})$ 
for some $\beta > 1$.
To obtain such an approximation, we later explain how to employ the primitive provided in \Cref{thm:mpc}.

\begin{algorithm}[H]
   \caption{Parallel clustering algorithm}
   \label{alg:offline_bi_criteria_clustering}
    \begin{algorithmic}[1]
    	\Require{parameter $\rho > 0$, point set $p \in P_w$ weighted by $w : P_w \to \mathbb{N}$}
    	\Require{primitive evaluates, given point $q\in P_w$ and radius $r>0$, the minimum function over 
    		a set $A^\beta_{P_w}(p, r)$ that approximates $B_{P_w}(p, r)$}
        \State initialize $\hat{C} \gets \emptyset$
        \State for each $p \in P_w$, let $w'(p) \gets w(p) + h(p)$ be the perturbed weight, where $h(p)\in [0, 1)$ is uniformly random
        	\label{ln:alg-clustering-parallel-perturb}
        \State for each $p \in P_w$
            \begin{itemize}
                \item[(C1)] add $p$ to $\hat{C}$ with probability $\mu / \gamma$ (independently of other points and of $h$)
                \item[(C2)] add $p$ to $\hat{C}$ if $p$ has the largest perturbed weight in a set $A^\beta_{P_w}(p, \sqrt[z]{\rho / w(p)})$ 
\end{itemize}
\end{algorithmic}
\end{algorithm}

\paragraph{Parallel Algorithm Analysis.}
We first bound $|\hat{C}|$.
In a nutshell, that the number of centers selected by (C2) is bounded by $(1+\mu)\cdot k$ follows directly 
from the analysis of the sequential algorithm. For (C1), the number of centers is just the sum of independent Bernoulli random variables.

\begin{lemma}\label{lemma:clustering-parallel-centers}
  $\E[|\hat{C}|] < (1 + 2\mu) k$ and
  $\Pr[ |\hat{C}| \geq (1 + 3\mu) k ] \leq \exp(-\mu k / 3)$.
\end{lemma}

\begin{proof}
	We first analyze the centers opened by (C2). The key observation is that if Algorithm~\ref{alg:offline_bi_criteria_clustering}
	opens a center at $p$, then the sequential algorithm (with tie-breaking given by the perturbed weights $w'(p)$)
	also opens a center at $p$. Indeed, this follows since if $p$ passes the test in (C2), then 
	it must have the largest weight in $A^\beta_{P_w}(p, \sqrt[z]{\rho / w(p)})\supseteq B_{P_w}(p, \sqrt[z]{\rho / w(p)})$ and thus,
	using the definition of $\rho = 2^z \cdot \OPTkzGuess / (\mu\cdot k)$,
	the sequential algorithm adds $p$ to $C$.
	(In more detail, note that all points $q\in B(p, \sqrt[z]{\rho / w(p)})$ satisfy 
	$\dist(p, q)^z\cdot w(p) \le \rho$ and $\rho$ is the threshold in the sequential algorithm.)
	Using Lemma~\ref{lemma:clustering-sequential-centers}, we get that (C2) opens less than $(1 + \mu)\cdot k$ centers
        (with probability $1$ because \Cref{lemma:clustering-sequential-centers} is deterministic).
	
	To analyze (C1), each of $\gamma \cdot k$ distinct points is selected with probability $\mu / \gamma$, 
	which together with the bound for (C2) implies $\E[|\hat{C}|] < (1 + 2\mu)\cdot k$.
	To get the high-probability bound, we apply the multiplicative Chernoff bound for a sum of independent Bernoulli random variables in a straightforward way.
\end{proof}

\begin{lemma}\label{lemma:clustering-parallel-connection_cost}
	$\E[\cl_z(P_w, \hat{C})] \le O(2^z\cdot \beta^z\cdot \gamma^2\cdot \OPTkz(P_w) / \mu^2)$.
\end{lemma}

\begin{proof}
	We show that the expected weighted connection cost of every point $p\in P_w$ is $\dist(p, \hat{C})^z \cdot w(p)\le O(\beta^z \cdot \gamma \cdot \rho / \mu)$,
	which implies the lemma by summing over all $\gamma \cdot k$ distinct points and using $\rho = 2^z \cdot \OPTkzGuess / (\mu\cdot k)\le 2^{z+1}\cdot \OPTkz(P_w) / (\mu\cdot k)$.
	
	To this end, similarly as in Algorithm~\ref{alg:def_seq}, consider the assignment sequence $p_0 = p, p_1, \dots, p_t$
	such that $p_t\in \hat{C}$ and for $i = 0, \dots, t-1$, $p_i\notin \hat{C}$ and $p_{i+1}$ is the point with the largest perturbed weight in $A^\beta_{P_w}(p_i, \rho_i)$,
	where we let $\rho_i := \sqrt[z]{\rho / w(p_i)}$ to simplify notation.
	We first observe that $w'(p_{i+1})\ge w'(p_i)$, by the definition of $p_{i+1}$.
	It follows that the weighted connection cost for $p$ is at most $t\cdot \beta^z \cdot \rho$ since for any $i < t$, we have that $A^\beta_{P_w}(p_i, \rho_i) \subseteq B(p_i, \beta\cdot \rho_i)$ and $w(p_{i+1})\ge w(p_i)$.
	It thus remains to show that $\E[t] \le \gamma / \mu$.
	
	We claim that for any $\ell \ge 0$, $\Pr[t > \ell] \le (1 - \mu / \gamma)^\ell$,
	which we prove by induction on $\ell$. For $\ell = 0$, the right-hand side is 1.
	For $\ell > 0$, we show that $\Pr[t > \ell \mid t > \ell - 1] \le (1 - \mu / \gamma)$, which clearly implies the inductive step. 
	Condition on $t > \ell - 1$ and 
	observe that if $A^\beta_{P_w}(p_\ell, \rho_\ell) \subseteq \bigcup_{i = 0}^{\ell - 1} A^\beta_{P_w}(p_i, \rho_i)$,
	then $p_\ell$ must have the largest weight in $ A^\beta_{P_w}(p_\ell, \rho_\ell)$, so it would be selected by (C2).
	Thus, for $t > \ell$, we must have that $A^\beta_{P_w}(p_\ell, \rho_\ell)$ contains at least one point $q\in P_w$
	that is not in $\bigcup_{i = 0}^{\ell - 1} A^\beta_{P_w}(p_i, \rho_i)$. 
	Point $q$ is selected by (C1) with probability $\mu / \gamma$, which implies that $\Pr[t > \ell \mid t > \ell - 1] \le (1 - \mu / \gamma)$.
	
	Finally, $\Pr[t > \ell] \le (1 - \mu / \gamma)^\ell$ implies that $\E[t] \le \gamma / \mu$ as the distribution of $t$ 
	is bounded by the geometric distribution.
\end{proof}

\paragraph{Guessing $\OPTkz(P_w)$.}
We run Algorithm~\ref{alg:offline_bi_criteria_clustering} in parallel with
$\OPTkzGuess = 2^i$ for every $i = 0, \ldots, \lfloor \log_2 (O(n\Delta))\rfloor$;
here we use that $1 \le \OPTkz(P_w)\le O(n \Delta)$.
We return the cheapest solution that uses at most $(1 + 3\mu)\cdot k$ centers.

Let $i$ be such that $\OPTkz(P_w)\le 2^i < 2\cdot \OPTkz(P_w)$.
We condition on that $|\hat{C}| < (1 + 3\mu)\cdot k$ when we run Algorithm~\ref{alg:offline_bi_criteria_clustering} with $\OPTkzGuess = 2^i$;
this holds with probability at least $1 - \exp(-\mu\cdot k / 3)$ by Lemma~\ref{lemma:clustering-parallel-centers}.
Thus, the algorithm will find a solution with less than $(1 + 3\mu)\cdot k$ centers.
Then Lemma~\ref{lemma:clustering-parallel-connection_cost} shows that the expected weighted clustering cost is
$O(2^z\cdot \beta\cdot \gamma^2\cdot \OPTkz(P_w) / \mu^2)$ for the value of $i$ chosen above,
and our algorithm may only return a cheaper solution. 
To get high-probability bounds, we run Algorithm~\ref{alg:offline_bi_criteria_clustering} with every $\OPTkzGuess = 2^i$ in parallel
$O(\log n)$ times and take the cheapest solution.

\paragraph{MPC Implementation of Algorithm~\ref{alg:offline_bi_criteria_clustering}.}
The MPC implementation of our parallel algorithm for weak coresets closely follows
Algorithm~\ref{alg:mpc_fac}. Perturbing the weights and performing rule (C1) is straightforward
as we do it at every point independently.

To verify the condition in (C2), we use \Cref{thm:mpc}.
Specifically, we first round the weight of every point to the nearest power of $2$, which increases the optimal cost of the coreset by at most a factor of $2$.
After that, there are at most $O(\log n)$ different weights.
Then we apply \Cref{thm:mpc} for every radius $r_i = \sqrt[z]{\rho / 2^i}$ for $i = 0, \dots, \lceil\log_2 n\rceil$
and function $f(S) = \max \{w'(p) : p\in S\}$ for $S \subseteq P_w$.
For a point $p\in P_w$ with $w(p) = 2^i$, we
denote the result of $f(A_P(p, r_i))$ as $\hat{w}'_{\max}(p)$,
and select $p$ as a center, i.e., add it to $\hat{C}$,
if $\hat{w}'_{\max}(p) = w'(p)$.

\medspace

\begin{proof}[Proof of Theorem~\ref{thm:clustering}]
Combining the MPC algorithm of Lemma~\ref{lemma:clustering-reduction_to_weighted} with the MPC implementation of Algorithm~\ref{alg:offline_bi_criteria_clustering} completes the proof of Theorem~\ref{thm:clustering}:
The number of rounds and total space follow from Lemma~\ref{lemma:clustering-reduction_to_weighted}, \Cref{thm:ufl},
 and \Cref{thm:mpc}.
 By Lemmas~\ref{lemma:clustering-parallel-connection_cost}
 and~\ref{lemma:clustering-reduction_to_weighted},
 the expected clustering cost is 
$O(2^z\cdot \beta^z\cdot \gamma^2\cdot \OPTkz(P_w) / \mu^2)
\le O(2^z\cdot \beta^z\cdot \gamma^3\cdot \OPTkz(P) / \mu^2)$;
thus only an $O(1)$-factor larger upper bound holds with constant probability, which we can amplify to $1 - 1/\poly(n)$ by running the algorithm $O(\log n)$-times independently in parallel.
Finally, we set $\gamma = \Gamma^{O(z)}$ from Theorem~\ref{thm:ufl}
and $\beta = O(\Gamma)$ by Theorem~\ref{thm:mpc}.
\end{proof}

\bibliographystyle{alphaurl}
  \bibliography{ref.bib}

\end{document}